\documentclass[11pt, draftcls, onecolumn]{IEEEtran}
\IEEEoverridecommandlockouts 
\usepackage{float}
\usepackage{graphicx}
\usepackage{subfigure}
\usepackage{color}
\usepackage{epsfig}
\usepackage{amssymb}
\usepackage{amsmath}
\usepackage{amsthm}
\usepackage{latexsym}
\usepackage{setspace}
\usepackage{bbm}
\usepackage{flushend}
\usepackage{dsfont}
\usepackage{xspace}
\usepackage[ruled,vlined]{algorithm2e}

\newcommand{\dP}{\mathrm{P}}
\newcommand{\dQ}{\mathrm{Q}}
\newcommand{\bPP}[1]{{\dP_{#1}}}
\newcommand{\bQQ}[1]{{\dQ_{#1}}}
\newcommand{\bPr}[1]{{\mathrm{P}}\left(#1\right)}

\newcommand{\bP}[2]{\mathrm{P}_{#1}\left({#2}\right)}
\newcommand{\bQ}[2]{\mathrm{Q}_{#1}\left({#2}\right)}

\newcommand{\unif}{\bPP {\mathtt{unif}}}
\newcommand{\mc}{-\!\!\!\!\circ\!\!\!\!-}

\newcommand{\cA}{{\mathcal A}}

\newcommand{\cE}{{\mathcal E}}
\newcommand{\cF}{{\mathcal F}}

\newcommand{\cI}{{\mathcal I}}

\newcommand{\cK}{{\mathcal K}}

\newcommand{\cT}{{\mathcal T}}

\newcommand{\cV}{{\mathcal V}}

\newcommand{\cX}{{\mathcal X}}
\newcommand{\cY}{{\mathcal Y}}
\newcommand{\cZ}{{\mathcal Z}}

\newcommand{\bF}{\mathbf{F}}

\newcommand{\ep}{\epsilon}
\newcommand{\la}{\lambda}

\usepackage{color}

\newtheorem{theorem}{Theorem}
\newtheorem{proposition}[theorem]{Proposition}
\newtheorem{corollary}[theorem]{Corollary}
\newtheorem*{corollary*}{Corollary}

\newtheorem{lemma}[theorem]{Lemma}
\newtheorem*{lemma*}{Lemmas}

\theoremstyle{remark}
\newtheorem*{remark*}{Remark}
\newtheorem*{remarks*}{Remarks}

\theoremstyle{definition}
\newtheorem{definition}{Definition}
\newtheorem{remark}{Remark}
\newtheorem{example}{Example}

\newenvironment{protocol}[1][htb]
  {%
   \begin{algorithm}[#1]%
  }{\end{algorithm}}

\newcommand{\ttlvrn}[2]{\left\| #1 - #2\right\|_1}

\newcommand{\espc}{\hspace*{-0.5cm}}

\newcommand{\indicator}{{\mathds{1}}}

\newcommand{\two}{{(2)}}
\newcommand{\m}{{(m)}}
\newcommand{\three}{{(3)}}

\newcommand{\slen}{{S_{\ep, \delta}}}
\newcommand{\sleng}{{\slen(X, Y \mid Z)}}

\newcommand{\ie}{{\it i.e.}\xspace}
\newcommand{\cf}{{\it cf.}\xspace}

\newcommand{\bX}{{\bf X}}
\newcommand{\bY}{{\bf Y}}
\newcommand{\bZ}{{\bf Z}}

\newcommand{\id}[2]{{i_{#1}\left(#2\right)}}
\newcommand{\idxy}{\id {XY}{x,y}}
\newcommand{\idXY}{\id {XY}{X,Y}}
\newcommand{\idxz}{\id {XZ}{x,z}}
\newcommand{\idXZ}{\id {XZ}{X,Z}}

\newcommand{\var}{\mathrm{\mathbb{V}ar}}

\newcommand{\hcolor}[1]{#1}

\begin{document}

\title{Secret Key Agreement: General Capacity and Second-Order Asymptotics}

\author{
\IEEEauthorblockN{Masahito Hayashi$^\ast$} 
\and
\IEEEauthorblockN{Himanshu Tyagi$^\dag$} 
\and
\IEEEauthorblockN{Shun Watanabe$^\ddag$} 
}

\maketitle

{\renewcommand{\thefootnote}{}\footnotetext{
%\hspace*{-.11in}\rule{24ex}{.05em}
\noindent$\ast$The Graduate School of Mathematics, Nagoya University, Japan,  and The Center for
Quantum Technologies, National University of Singapore, Singapore. Email:masahito@math.nagoya-u.ac.jp

\noindent$\ast$Department of Electrical Communication Engineering, Indian
Institute of Science, Bangalore 560012, India. 
Email: htyagi@ece.iisc.ernet.in

\noindent$^\ddag$Department of Computer and Information Sciences, 
Tokyo University of Agriculture and Technology, Tokyo 184-8588, Japan. 
Email: shunwata@cc.tuat.ac.jp

An initial version of this paper was presented at the IEEE International Symposium on Information Theory,  Hawaii, USA, 2014.
}}

\maketitle

\renewcommand{\thefootnote}{\arabic{footnote}}
\setcounter{footnote}{0}

\begin{abstract}
We revisit the problem of secret key agreement 
using interactive public communication for  
two parties and propose a new secret key 
agreement protocol. The protocol attains the secret key
capacity for general observations and attains 
the second-order asymptotic term in the maximum
length of a secret key for independent and identically 
distributed observations.
In contrast to the previously suggested 
secret key agreement protocols, the proposed protocol
uses interactive communication.
In fact, the standard one-way communication protocol
used prior to this work fails to attain the asymptotic 
results above.
Our converse proofs rely on a recently established 
upper bound for secret key lengths. 
Both our lower and upper bounds are derived in 
a {\it single-shot} setup and the asymptotic
results are obtained as corollaries.
 
\end{abstract}
%%%%%%%%%%%%%%%%%%%%%%%%%%%%%%%%%%%%%%%%%%%%%%%
\section{Introduction}\label{s:introduction}
Two parties observing {\it random variables} (RVs) $X$ and $Y$
seek to agree on a secret key. They can communicate interactively
over an error-free, authenticated, albeit insecure, communication channel of 
unlimited capacity. The secret key must be concealed from an eavesdropper with access to the 
communication and an additional side information $Z$.
What is the maximum length $S(X, Y\mid Z)$ of a secret key that the parties can agree upon?

A study of this question was initiated by Maurer \cite{Mau93} and Ahlswede and Csisz{\'ar}
\cite{AhlCsi93} for the case where the observations of the parties and the eavesdropper 
consist of $n$ {\it independent
and identically distributed} (IID) repetitions $(X^n, Y^n, Z^n)$ of RVs $(X, Y, Z)$. For the case
when $X\mc Y \mc Z$ form a Markov chain, it was shown in \cite{Mau93, AhlCsi93}
that the {\it secret key capacity} equals $I(X\wedge Y|Z)$, namely
$$S(X^n, Y^n\mid Z^n ) = nI(X\wedge Y\mid Z) + o(n).$$
However, in several applications (see, for instance, \cite{DodOstReySmi08})
the observed data is not IID or even if the observations are IID, 
the observation length $n$ is limited and a more precise asymptotic analysis 
is needed.  

In this paper, we address the secret key agreement problem for these two important
practical situations. First, when the observations consist of general
sources (\cf \cite{HanVer93, Han03}) $(X_n, Y_n, Z_n)$ such that $X_n \mc Y_n \mc Z_n$ is a Markov chain, 
we show that
$$S(X_n, Y_n\mid Z_n) = n\underline{I}({\bf X}\wedge {\bf Y}\mid \bZ) + o(n),$$
where $\underline{I}({\bf X}\wedge {\bf Y}\mid \bZ)$ is the {\it inf-conditional information} of
${\bf X}$ and ${\bf Y}$ given $\bZ$. Next, for the IID case with $X\mc Y\mc Z$,
we  identify the second-order asymptotic term\footnote{Following the pioneering work of Strassen 
\cite{Str62}, study of these second-order terms in coding theorems has been revived 
recently by Hayashi \cite{Hayashi08, Hay09} and Polyanskiy, Poor, and 
Verd\'u~\cite{PolPooVer10}.} in $S(X^n, Y^n\mid Z^n)$.
Specifically, denoting by $S_{\ep, \delta}(X, Y\mid Z)$ the maximum length 
of a secret key over which the parties agree with probability greater than $1 - \ep$ and 
with secrecy parameter less than $\delta$, we show that
$$S_{\ep, \delta}\left(X^n,Y^n\mid Z^n\right) = nI(X\wedge Y\mid Z) - \sqrt{nV}Q^{-1}(\ep+\delta) \pm {\cal O}(\log n),$$
where $Q$ is the tail probability of the standard Gaussian distribution and
$$V := \var\left[ \log \frac{\bP{XY\mid Z}{X,Y\mid Z}}{\bP{X\mid Z}{X\mid Z} \bP{Y\mid Z}{Y\mid Z}} \right].$$
In particular, our bounds allow us to evaluate the {\it gap to secret key capacity}
at a finite blocklength $n$. In Figure~\ref{f:gap_to_capacity} we illustrate 
this gap between the maximum possible rate of a secret key at a fixed $n$ 
and the secret key capacity
for the case where $Z$ is a random bit, $Y$ is obtained by flipping 
$Z$ with probability $0.25$ and $X$ given by flipping $Y$ with probability
$0.125$; see Example~\ref{e:gap_to_capacity} in Section~\ref{s:second_order} for details.
\begin{figure}[h]
\begin{center}
\includegraphics[scale=0.5]{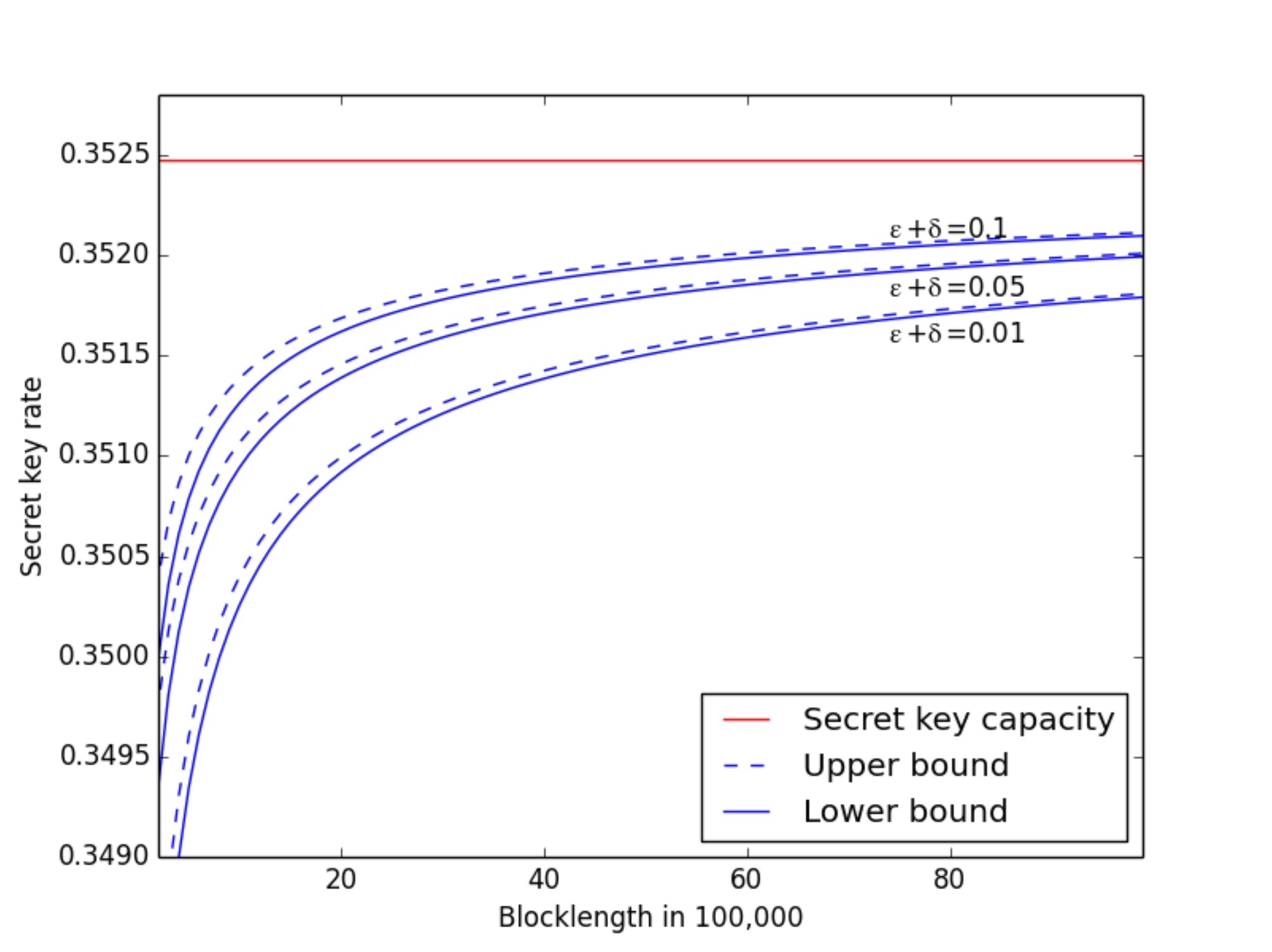}
\caption{Gap to secret key capacity at finite $n$
for $\ep+\delta = 0.01$, $0.05$, $0.1$.}
\label{f:gap_to_capacity}
\end{center}
\end{figure}

Underlying these results is a general single-shot characterization of 
the secret key length which shows that, when $X\mc Y\mc Z$, $S_{\ep,\delta}(X,Y|Z)$
roughly equals the $(\ep+\delta)$-tail of the random variable
\[
i(X\wedge Y|Z) = \log \frac{\bP{XY|Z}{X,Y|Z}}{\bP {X|Z}{X|Z}\bP{Y|Z}{Y|Z}}.
\]
Our main technical contribution in proving this result 
is a new single-shot secret key agreement protocol
which uses interactive communication and attains the desired optimal performance. 
It was observed in \cite{Mau93, AhlCsi93} that
a simple one-way communication protocol suffices to attain the secret key capacity,
when the Markov relation $X \mc Y \mc Z$ holds. 
Also, for a multiterminal setup with constant $Z$, 
\cite{CsiNar04} showed
that a noninteractive communication protocol achieves the secret key capacity. 
Prior to this work,  when the Markov relation $X \mc Y \mc Z$ holds,
only such noninteractive communication protocols were used for generating
secret keys\footnote{Interaction is known to help 
in some cases where neither $X \mc Y \mc Z$ nor $Y \mc X \mc Z$ is satisfied
\cite{Mau93,WatMatUye07,GohAna10}.}, even in single-shot setups (\cf \cite{RenWol05}).
In contrast, our proposed protocol uses interactive communication.
We note in Remark~\ref{r:noninteractive_insufficient} that none of the standard one-way
communication protocols achieve the optimal asymptotic bounds,
suggesting that perhaps interaction is necessary
for generating a secret key of optimal length (see Section~\ref{s:discussion}
for further discussion and an illustrative example).
%In particular, the single-shot protocol suggested in \cite{RenWol05}, which attains
%the secret key capacity in the asymptotic regime, does not yield the claimed asymptotic 
%bounds on secret key lengths. 

Typically, secret key agreement protocols consist of two steps:
{\it information reconciliation} and {\it privacy amplification}. 
In the first step, the parties communicate to generate some shared 
random bits, termed {\it common randomness}. However, the communication
used leaks some information about the generated common randomness. Therefore,
a second privacy amplification step is employed to extract 
from the common randomness secure random
bits that are almost independent of the communication used. 
For IID observations,
the information reconciliation step of 
the standard one-way secret key agreement protocol 
entails the two parties agreeing on $X$ using a 
one-way communication of rate $H(X|Y)$. In the privacy amplification step,
the rate $H(X|Z)$ residual randomness of $X$, which is almost independent of $Z$, is 
used to extract a secret key of rate $H(X|Z) - H(X|Y)$ which is independent jointly
of $Z$ and the communication used in the information reconciliation stage.
Under the Markov condition $X \mc Y \mc Z$, the resulting secret key attains 
the secret key capacity $I(X \wedge Y \mid Z)$. 
However, in the single-shot regime, the behavior of RVs
$- \log \bP{X|Z}{X|Z}$ and $-\log \bP{X|Y}{X|Y}$, rather than their expected values, becomes relevant. 
The difficulty in extending the standard one-way communication protocol
to the single-shot setup lies in the spread of the {\it information spectrums}\footnote{The range of the 
log-likelihood $-\log \bP X x $ (conditional log-likelihood $-\log \bP {X|Y}{x|y}$)
is referred to as information spectrum of $\bPP X$ (conditional information spectrum of
$\bPP{X|Y}$). This notion was introduced in the seminal 
work \cite{HanVer93} and is appropriate for deriving single-shot coding theorems, without making
assumptions on the underlying distribution. See \cite{Han03} for a detailed account.}
of $\bPP{X|Y}$ and $\bPP{X|Z}$. Specifically, 
while we require the random variable $-\log \bP{X|Y}{X|Y}$ itself to show up as the length
of communication in the information reconciliation step, a naive extension requires
as much communication as a large probability tail of $-\log \bP{X|Y}{X|Y}$.
To remedy this, we {\it slice the spectrum} of $\bPP{X|Y}$ into slices\footnote{
See Appendix~\ref{appendix:high_secrecy_protocol} for a secret key agreement
based on slicing the spectrum of $\bPP X$.}
of length $\Delta$ each and adapt the protocol to the slice which contains
$(X,Y)$. However, since neither party knows the value of
$-\log \bP{X|Y}{X|Y}$, this adaptation requires interactive communication.

Motivating this work, and underlying our converse proof, is a recently established 
single-shot upper bound on secret key lengths for the multiparty secret key agreement problem \cite{TyaWat14} (see, also, \cite{TyaNar13ii}).
The proof relies on relating secret key agreement to binary hypothesis 
testing. In spirit, this result can be regarded as a multiterminal variant 
of a similar single-shot converse for the channel coding problem
which appeared first in \cite{Nagaoka01, HayNag03} and has been termed
the meta-converse by Polyanskiy, Poor, and Verd\'u~\cite{PolPooVer10, Pol10}
(see, also, \cite{WanRen12} and \cite[Section 4.6]{Hayashi06}).

The basic concepts of secret key agreement and a general result 
for converting a high reliability protocol to a high secrecy protocol 
are given in the next section. In Section \ref{s:single_shot_upper_bound},
we review the single-shot upper bound of \cite{TyaWat14} for the two party case. 
Our new secret key agreement protocol and its single-shot performance analysis is
presented in Section \ref{s:protocols}. The single-shot results are applied to general
sources in Section \ref{s:general_sources} and to IID sources in Section \ref{s:second_order}. 
The final section contains a discussion on the role of 
interaction in our secret key agreement protocols.

%%%%%%%%%%%%%%%%%%%%%%%%%%%%%%%%%%%%%%%%%%%%%%%
\section{Secret keys}\label{s:secret_keys}
We consider the problem of secret key agreement using interactive public 
communication by two (trusted) parties observing
RVs $X$ and $Y$ taking values in countable sets $\cX$ and $\cY$, respectively.
Upon making these observations, the parties communicate interactively over 
a public communication channel that is accessible by an eavesdropper. We 
assume that the communication channel is error-free and authenticated. 
Specifically, the communication is sent over $r$ rounds of interaction\footnote{
In the asymptotic regime considered in Sections~\ref{s:general_sources}, 
the number of rounds $r$ may depend on the block length $n$.}. In 
the $j$th round of communication, $1\leq j \leq r$, each party sends a
message which is a function of its observation, {\it locally generated} 
randomness denoted by\footnote{The RVs $U_x$ and $U_y$
are mutually independent and independent jointly of $(X, Y)$.}
$U_x$ and $U_y$, and the 
previously observed communication. The overall interactive communication is 
denoted by $\bF$. In addition to $\bF$, the eavesdropper observes a RV
$Z$ taking values in a countable set $\cZ$. The joint distribution $\bPP{XYZ}$
is known to all parties.

Using the interactive communication $\bF$ and their local observations, 
the parties agree on a secret key. A RV $K$ constitutes a secret key if the two parties form
estimates that agree with $K$ with
probability close to $1$ and $K$ is  concealed, in effect, from an eavesdropper
with access to $(\bF, Z)$. 
Formally, we have the following definition.

\begin{definition}\label{d:secret_key}
A RV $K$ with range $\cK$ constitutes 
an {\it$(\ep, \delta)$-secret key} ($(\ep, \delta)$-SK) if there exist functions $K_x$ and $K_y$ of $(U_x, X, \bF)$ and $(U_y, Y, \bF)$,
respectively, such that
the following two conditions are satisfied 
\begin{align}
\bPr{K_x =  K_y = K} &\geq 1 - \ep, 
\\
\ttlvrn {\bPP{K\bF Z}} {\bPP{\mathtt{unif}}\times \bPP{\bF Z}} &\leq \delta,
\label{eq:secrecy-condition}
\end{align}
where $\bPP{\mathtt{unif}}$ is the uniform 
distribution on $\cK$ and 
$$\ttlvrn{\mathrm{P}}{\mathrm{Q}} = \frac 12\sum_u |\mathrm{P}(u) - \mathrm{Q}(u)|.$$
\end{definition}
The first condition above represents the {\it reliability} 
of the secret key and the second condition guarantees {\it secrecy}. 

\begin{definition}\label{d:secret_key_length}
Given $\ep, \delta \in [0,1)$,  the supremum over the lengths $\log|\cK|$ of an $(\ep, \delta)$-SK is denoted by 
$S_{\ep, \delta}(X, Y \mid Z)$. 
\end{definition}

\begin{remark}\label{r:interesting_ep_delta}
The only interesting case is when $\ep+\delta<1$, since otherwise
$S_{\ep, \delta}(X, Y \mid Z)$ is unbounded. Indeed, consider two trivial
secret keys $K_1$ and $K_2$ with range $\cK$ 
generated as follows: For $K_1$, the first party generates $K_x=K_1$ uniformly over
$\cK$ and sends it to the second party. Thus, $K_1$ constitutes a $(0, 1 - 1/|\cK|)$-SK,
and therefore, also a $(0,1)$-SK. For $K_2$, the first party generates $K_x=K_2$ uniformly
over $\cK$ and the second party generates $K_y$ uniformly over $\cK$. Then, $K_2$ constitutes
a $(1-1/|\cK|, 0)$-SK, and therefore, also a $(1,0)$-SK.
If $\ep+\delta \ge 1$, the RV $K$ which equals $K_1$ with probability $(1-\ep)$ and $K_2$ 
with probability $\ep$
constitutes $(\ep, 1-\ep)$-SK of length $\log|\cK|$, 
and therefore, also an $(\ep, \delta)$-SK of the same length. Since $\cK$ was arbitrary, 
$S_{\ep, \delta}(X, Y \mid Z) = \infty$.
\end{remark}
Remark \ref{r:interesting_ep_delta} exhibits 
a high reliability $(0,1)$-SK and a high secrecy $(1,0)$-SK
for the trivial case $\ep+\delta \ge 1$. The two constructions 
together sufficed to characterize $\sleng$. Following a similar approach for the regime 
$\ep+\delta<1$,
we can construct a high reliability $(\ep+\delta, 0)$-SK and a high secrecy
$(0, \ep+\delta)$-SK and 
randomize over those two secret keys with probabilities $\ep/(\ep+\delta)$ and $\delta/(\ep+\delta)$
to obtain a hybrid, $(\ep, \delta)$-SK. However, the results below show that we do not need 
to construct both  high secrecy and high reliability secret keys for the secrecy definition in 
\eqref{eq:secrecy-condition}
and a high reliability construction alone will suffice. We first show that
 any $(\ep, \delta)$-SK can be converted into
a high secrecy, $(\ep+\delta, 0)$-SK.

\begin{proposition}[{{\bf Conversion to High Secrecy Protocol}}] \label{p:conversion}
Given an $(\ep, \delta)$-SK, there exists an $(\ep+\delta,0)$-SK of the same length.
\end{proposition}
{\it Proof.} Let $K$ be an $(\ep, \delta)$-SK using interactive communication $\bF$, with local estimates $K_x$ and $K_y$.
\hcolor{
We construct a new $(\ep+\delta, 0)$-SK $K^\prime$ using the {\it maximal coupling lemma}, which asserts the following (\cf \cite{Strassen65}):
Given two distributions $\dP$ and $\dQ$ on a set $\cX$, there exists a joint distribution $\bPP{XX^\prime}$
on $\cX\times \cX$ such that the marginals are $\bPP{X} = \dP$ and $\bPP{X^\prime} = \dQ$, and under $\bPP{XX^\prime}$
\begin{align}
\bPr{X\neq X^\prime} = \ttlvrn \dP \dQ.
\label{e:maximal_coupling_property}
\end{align}
The distribution $\bPP{XX^\prime}$ is called the {\it maximal coupling} of $\dP$ and $\dQ$.
}

For each fixed realization of $(\bF, Z)$, let 
$\bPP{KK^\prime\mid \bF, Z}$ be the {maximal coupling}
of $\bPP{K\mid \bF Z}$ and $\unif$. Then
$\bPP{K^\prime\bF Z} = \unif \times \bPP{\bF Z}$, and since $K$ is an $(\ep, \delta)$-SK, we get by the maximal coupling property
\eqref{e:maximal_coupling_property} that
\[
\bPr{K\neq K^\prime} =  \ttlvrn{\bPP{K\bF Z}}{\unif\times \bPP{\bF Z}}\le \delta,
\]
\hcolor{
and define 
\begin{align}
\bPP{K^\prime KK_xK_y\bF XYZU_xU_y} := \bPP{K^\prime|K\bF Z} \bPP{KK_xK_y\bF XYZU_xU_y}.
\label{e:joint_distribution_coupling}
\end{align}
Since $\bPr{K = K_x = K_y} \geq 1-\ep$, under $\bPP{K^\prime KK_xK_y\bF XYZ}$ we have
\[
\bPr{K_x = K_y = K^\prime}\geq 1- \ep -\delta.
\]
}
Thus, $K^\prime$ constitutes an $(\ep+\delta,0)$-SK. \qed

\hcolor{
Proposition~\ref{p:conversion} plays an important role in 
our secret  key agreement protocol and allows us to convert 
a high reliability $(\eta, \alpha)$-SK
with small $\eta$ into an $(\ep, \delta)$-SK for any
arbitrary $\ep$ and $\delta$
satisfying (roughly) $\ep+\delta> \alpha$. Formally, we have the following.
}
\begin{proposition}[{{\bf Hybrid Protocol}}] \label{p:time_sharing}
Given a protocol for generating $(\eta,\alpha)$-SK, there exists
a protocol for generating an $(\ep, \delta)$-SK of the same
length for every $0< \ep, \delta<1$ such that
\begin{align*}
\ep &\ge \eta\\
\ep + \delta &\ge \alpha + \eta.
\end{align*}
\end{proposition}
{\it Proof.} Given an $(\eta, \alpha)$-SK $K_1$, by 
Proposition~\ref{p:conversion} there exists an 
$(\alpha+\eta, 0)$-SK $K_2$. 
Let $\theta = \delta/(\ep-\eta + \delta)$. Consider
a secret key $K$ obtained by a hybrid use of the
protocols for generating $K_1$ and $K_2$, with the protocol
for $K_1$ executed with probability $\theta$ and 
that for $K_2$ with probability $1-\theta$.
Note from the proof of Proposition~\ref{p:conversion}
that it is 
the same secret key agreement protocol $(K_x, K_y, \bF)$ 
that generates both $K_1$ and $K_2$.
Thus, the claim follows for the time-shared secret key $K$
since
\begin{align}
\bPr{K = K_x = K_y}
&=\theta\bPr{K_1=K_x=K_y}+(1-\theta)\bPr{K_2=K_x=K_y}
\nonumber
\\
&\geq 1 - \frac{\delta \eta + (\ep - \eta)(\alpha+\eta)}{\ep -\eta + \delta}
\nonumber
\\
&\geq 1 - \frac{\delta \eta + (\ep - \eta)(\ep+\delta)}{\ep -\eta + \delta}
\label{e:ep+delta_use1}
\\
&= 1- \ep,
\nonumber
\end{align}
and
\begin{align}
\ttlvrn{\bPP{K\bF Z}}{\bPP{\mathtt{unif}}\bPP{\bF Z}}
&\leq \theta\ttlvrn{\bPP{K_1\bF Z}}{\bPP{\mathtt{unif}}\bPP{\bF Z}}
+(1-\theta)\ttlvrn{\bPP{K_2\bF Z}}{\bPP{\mathtt{unif}}\bPP{\bF Z}}
\nonumber
\\
&\leq \frac{\delta \alpha + (\ep-\eta)\cdot 0}{\ep -\eta + \delta}
\nonumber
\\
&\leq \frac{\delta (\ep+\delta - \eta)}{\ep -\eta + \delta}
\label{e:ep+delta_use2}
\\
&= \delta,
\nonumber
\end{align}
where we have used the assumption $\ep+\delta\geq \alpha+\eta$ in 
\eqref{e:ep+delta_use1} and \eqref{e:ep+delta_use2}.
\qed

\begin{remark}\label{r:actual_key}
Note that the actual secret key $K$ in Definition~\ref{d:secret_key}
is not available to any party
and has only a formal role in the secret key agreement protocol. Interestingly,
the proof above  says that the estimates $(K_x,K_y)$ of a high reliability $(\eta, \alpha)$-SK
with $\eta \approx 0$ constitute an $(\ep,\delta)$-SK as well for every $\ep+\delta\gtrsim \alpha$,
 albeit for a different hidden RVs $K$.
\end{remark}

Thus, it suffices to exhibit a high reliability 
 $(\eta,\alpha)$-SK with
small $\eta$ and desired $\alpha$. Such a protocol
is given in Section~\ref{s:protocols} and underlies
all our achievability results.

{\bf A more demanding secrecy requirement.} A more demanding secrecy requirement enforces one of the estimates $K_x$ or 
$K_y$ itself to be secure, \ie, 
\begin{align}
\ttlvrn{ \bPP{K_x\bF Z}} { \bPP{\mathtt{unif}} \bPP{\bF Z}} \leq \delta
\text { or } 
\ttlvrn{\bPP{K_y\bF Z}} { \bPP{\mathtt{unif}} \bPP{\bF Z}} \leq \delta.
\label{eq:secrecy-condition-demanding}
\end{align}
The validity of Proposition~\ref{p:conversion} for this secrecy requirement
remains open. However, in the important special case when $Z$ 
is a function
of either $X$ or $Y$, which includes the case of constant $Z$, Proposition~\ref{p:conversion}
holds even under the more demanding secrecy requirement \eqref{eq:secrecy-condition-demanding}.
Indeed, let $(K_x, K_y)$ be an $(\ep, \delta)$-SK with $K_x$
satisfying the more demanding secrecy requirement above.
Proceeding as in the proof of Proposition~\ref{p:conversion}
with $K_x$ in the role of $K$, we obtain a RV 
$K^\prime$ such that $\bPr{K^\prime \neq K_y}\leq \ep+\delta$
and $\bPP{K^\prime \bF Z} = \bPP{\mathtt{unif}}\times \bPP{\bF Z}$.
Let the joint distribution $\bPP{K^\prime K_xK_y\bF XYZU_xU_y}$
be as in \eqref{e:joint_distribution_coupling} with $K_x$ replacing
$K$. To claim that $K^\prime$ constitutes an $(\ep+\delta,0)$-SK
under \eqref{eq:secrecy-condition-demanding}, it suffices to show
that one of the parties can simulate $K^\prime$.
To that end, the party observing $X$ can first run the original secret key agreement protocol
to get $K_x$ and $\bF$. Also, $Z$ is available to the party since it is
a function of $X$. Thus, this party can simulate the required RV $K^\prime$
using the distribution $\bPP{K^\prime|K_x\bF Z}$, which completes the proof
of Proposition~\ref{p:conversion} under the more demanding secrecy requirement.

Note that the argument above relies on using local randomness to simulate $K^\prime$.
For the original secrecy requirement \eqref{eq:secrecy-condition}, 
Proposition~\ref{p:conversion} holds even when we restrict to deterministic 
protocols with no local randomness allowed. It turns out that this is not
the case for the more demanding secrecy requirement \eqref{eq:secrecy-condition-demanding}, 
as the following simple
counterexample shows: Let $X$ be a binary RV taking $1$ with probability $p < \frac{1}{2}$,
and let $Y = Z =$ constant. Then, $K_x = X$ and constitutes a $1$-bit $(p, 1/2-p)$-SK
under \eqref{eq:secrecy-condition-demanding}. If Proposition~\ref{p:conversion}
holds, the parties should be able to generate a $(1/2, 0)$-SK. However,
a $(1/2,0)$-SK consists of an unbiased bit, which cannot be generated 
without additional randomness. Therefore, for secrecy requirement \eqref{eq:secrecy-condition-demanding}, 
Proposition~\ref{p:conversion} does not hold if we restrict to deterministic protocols.

To conclude, for the special case when $Z$ is a function of $X$,
it suffices to construct only a high reliability protocol,
provided that local randomness is available. In fact, the high reliability
protocol proposed in Section~\ref{s:protocols} satisfies the more 
demanding secrecy requirement \eqref{eq:secrecy-condition-demanding}
and, if $Z$ is a function of either $X$ or $Y$, 
all the results of this paper hold even under \eqref{eq:secrecy-condition-demanding}.

%In the following two sections we shall present upper and lower bounds on 
%$S_{\ep, \delta}(X, Y \mid Z)$. Subsequently, for the case when $X \mc Y \mc Z$ form a 
%Markov chain, these bounds will be used to establish the secret key capacity for general 
%sources and the second-order asymptotics of secret key rates when the underlying observations 
%$(X, Y, Z)\equiv (X^n, Y^n, Z^n)$ are IID.

%%%%%%%%%%%%%%%%%%%%%%%%%%%%%%%%%%%%%%%%%%%%%%%
\section{Upper bound on $S_{\ep, \delta}(X, Y \mid Z)$}\label{s:single_shot_upper_bound}
We recall the {\it conditional independence testing} upper bound on $S_{\ep, \delta}(X, Y \mid Z)$,
which was established recently in \cite{TyaWat14, TyaWat14ii}. In fact, the general upper bound 
in \cite{TyaWat14, TyaWat14ii} is a single-shot upper bound
on the secret key length for a multiparty secret key agreement problem. We recall
a specialization of the general result to the case at hand. In order
to state our result, we need the following concept from binary hypothesis
testing.

Consider a binary hypothesis testing problem with null hypothesis $\mathrm{P}$ and alternative hypothesis $\mathrm{Q}$, where $\mathrm{P}$ and $\mathrm{Q}$ are distributions on the same alphabet ${\cal V}$. Upon observing a value $v\in \cV$, the observer needs to decide if the value was generated by the distribution $\bPP{}$ or the distribution $\mathrm{Q}$. To this end, the observer applies a stochastic test $\mathrm{T}$, which is a conditional distribution on $\{0,1\}$ given an observation $v\in \cV$. When $v\in \cV$ is observed, the test $\mathrm{T}$ chooses the null hypothesis with probability $\mathrm{T}(0|v)$ and the alternative hypothesis with probability $T(1|v) = 1 - T(0|v)$. For $0\leq \ep<1$, denote by $\beta_\ep(\mathrm{P},\mathrm{Q})$ the infimum of the probability of error of type II given that the probability of error of type I is less than $\ep$, \ie,
\begin{eqnarray}
\beta_\ep(\mathrm{P},\mathrm{Q}) := \inf_{\mathrm{T}\, :\, \mathrm{P}[\mathrm{T}] \ge 1 - \ep} \mathrm{Q}[\mathrm{T}],
\label{e:beta-epsilon}
\end{eqnarray}
where 
\begin{eqnarray*}
\mathrm{P}[\mathrm{T}] &=& \sum_v \mathrm{P}(v) \mathrm{T}(0|v), \\
\mathrm{Q}[\mathrm{T}] &=& \sum_v \mathrm{Q}(v) \mathrm{T}(0|v).
\end{eqnarray*}

%The following upper bound is a special case of the bound established in \cite{TyaWat14, TyaWat14ii}. 
The definition of a secret key used in  \cite{TyaWat14, TyaWat14ii} is different from 
Definition~\ref{d:secret_key}. However, the two definitions are closely related, and the
upper bound of \cite{TyaWat14, TyaWat14ii} can be extended to our case as well.
We review the alternative definition in Appendix~\ref{appendix:theorem:one-shot-converse-source-model} 
and relate it to Definition~\ref{d:secret_key} to derive the following upper bound, which
will be instrumental in our converse proofs.
\begin{theorem}[{{\bf Conditional independence testing bound}}] \label{theorem:one-shot-converse-source-model}
Given $0\leq \ep+ \delta <1$, $0<\eta<1-\ep-\delta$, the following bound holds:
\begin{eqnarray}
S_{\ep, \delta}\left(X, Y \mid  Z\right)
\le -\log \beta_{\ep+\delta+\eta}\big(\bPP{XYZ},\mathrm{Q}_{X|Z}\mathrm{Q}_{Y|Z}\mathrm{Q}_{Z}\big)  + 2 \log(1/\eta),
\nonumber
\end{eqnarray}
for all joint distributions $\mathrm{Q}$ on $\cX\times \cY\times \cZ$ that render $X$ and $Y$ conditionally independent given $Z$. 
\end{theorem}

%%%%%%%%%%%%%%%%%%%%%%%%%%%%%%%%%%%%%%%%%%%%%%%
\section{The secret key agreement protocol}\label{s:protocols}

In this section we present our secret key agreement protocol, which will
be used in all the achievability results of this paper.
A typical secret key agreement protocol for two parties entails sharing the observations
of one of the parties, referred to as {\it information reconciliation},
and then extracting a secret key out of the shared observations,
referred to as {\it privacy amplification} 
(\cf \cite{Mau93}, \cite{AhlCsi93}, \cite{Ren05}). In another interpretation,
the parties communicate first to establish a {\it common randomness}
\cite{AhlCsi98} and then extract a secret key from the common randomness\footnote{For an interpretation
of secrecy agreement in terms of common randomness decomposition, see \cite{CsiNar04, CsiNar08, TyaThesis}.}.
Our protocol below, too,
has these two components but the rate of the communication for information
reconciliation and the rate of the randomness extracted by privacy amplification
have to be chosen carefully. 

Heuristically, in the information reconciliation stage, the first party randomly bins $X$ and
sends it to the second party. 
If we do not use interaction, by the Slepian-Wolf theorem \cite{SleWol73} (see \cite{MiyKan95},\cite[Lemma 7.2.1]{Han03}, \cite{Kuz12}
for a single-shot version) the length of communication needed is roughly equals
a large probability upper bound for $h(X|Y) := -\log \bP{X|Y}{X|Y}$. 
However, in order to derive a lower bound that matches our upper bound, we expect 
to send $X$ to $Y$ using approximately $h(X|Y)$ bits of communication, which
can differ from the tail bound above by as much as the length of the spectrum of 
$\bPP{X|Y}$.
%On the other hand, the Markov condition $X\mc Y \mc Z$
%and a manipulation of terms shows that the number of secret key bits that can be extracted
%is related to the large probability lower for $h(X|Y)$. As a consequence, we get a slack 
%in our performance bound proportional to the length of the  spectrum of $\bPP{X|Y}$.
To overcome this gap, we utilize {spectrum slicing}, a technique introduced 
in \cite{Han03}, to construct an adaptive scheme that can handle the 
{\it spread of information spectrum}. 
Specifically, we divide the spectrum of $\bPP{X|Y}$ into slices of small lengths.
The protocol proceeds interactively to adapt to the current slice index, allowing
us to replace the spectrum length in the argument above with the length of a single slice.

\subsection{Formal description of the protocol}
We consider the essential spectrum of $\bPP{X|Y}$, \ie, the set of values taken by
$h(X|Y)$ between $\la_{\min}$ and $\la_{\max}$,
where $\la_{\min}$ and $\la_{\max}$ are chosen such that
$h(X|Y)$ lies in $(\la_{\min}, \la_{\max})$  with large probability.
We divide the essential spectrum of $\bPP{X|Y}$ 
into $L$ slices, $L$ of them of width $\Delta$. Specifically, for $1\leq j \leq L$, 
the $j$th slice of the spectrum of $\bPP{X|Y}$ is defined as follows
\begin{align*}
\cT_j = \{(x,y) : \lambda_j \leq - \log \bP{X|Y}{x|y} < \lambda_j + \Delta\},
\end{align*}
where $\lambda_j = \lambda_{\min} + (j-1)\Delta$. Note that
 the slice index is not available to any one party. The proposed
protocol proceeds assuming the lowest index $j=1$ and uses
interactive communication to adapt to the actual slice index.

For information reconciliation, we simply send a random binning of $X$.
However, the bin size is increased successively, where the incremental bin sizes $M_1,\ldots,M_L$ are given by 
\[
\log M_j = \begin{cases}
\lambda_1 +\Delta + \gamma, &\quad j = 1 \\
\Delta, &\quad 1 < j \leq L,
\end{cases}
\]

For privacy amplification, we will rely on the {\it leftover hash lemma} \cite{ImpLevLub89, RenWol05}. Let $\cF$ be a {\it $2$-universal family} of mappings $f: \cX\rightarrow \cK$, \ie, for each $x'\neq x$, the family $\cF$ satisfies 
\begin{align} \label{eq:property-two-universal}
\frac{1}{|\cF|} \sum_{f\in \cF} \indicator(f(x) = f(x')) \leq \frac{1}{|\cK|}.
\end{align}
The following lemma is a slight modification of the known forms of the leftover hash lemma (\cf~\cite{Ren05}); 
we give a proof in Appendix \ref{appendix:l:leftover_hash} for completeness.
\begin{lemma}[{{\bf Leftover Hash}}]\label{l:leftover_hash} Consider RVs $X$, $Z$, and $V$ taking values in
$\cX$, $\cZ$ and $\cV$, respectively, where $\cX$ and $\cZ$ are countable and $\cV$ is finite.
Let $S$ be a random seed such that $f_S$ is uniformly distributed over a $2$-universal family as above. 
Then, for $K = f_S(X)$ and for any $\bQQ{Z}$ satisfying $\mathrm{supp}(\bPP{Z}) \subset \mathrm{supp}(\bQQ{Z})$, we have
\begin{align*}
\ttlvrn{\bPP{KVZS}}{\bPP{\mathtt{unif}}\bPP{VZ}\bPP{S}} \leq \frac{1}{2} \sqrt{|\cK||\cV|2^{- H_{\min} \left(\bPP{XZ}\mid \bQQ{Z}\right)}},
\end{align*}
where $\bPP{\mathtt{unif}}$ is the uniform distribution on $\cK$ and 
$$H_{\min} \left(\bPP{XZ}\mid \bQQ{Z}\right) = - \log \sup_{x,z: \bQ{Z}{z}>0} \frac{\bP{XZ}{x,z}}{\bQ{Z}{z}}.$$
\end{lemma}
The main benefit of the spectrum slicing approach above is
that roughly $h(X|Y)+L+\Delta$ bits are sent
for each realization $(X,Y)$.
At the same time, we can estimate $h(X|Y)$ up to a precision of 
$\Delta$ when the protocol stops -- a key property in our secrecy analysis.

The complete protocol is described in Protocol~\ref{prot:high_reliability}.
We remark that the random seed based secret key generation used in 
Protocol~\ref{prot:high_reliability} is only for the ease of security analysis. A slight modification
of our protocol can work with deterministic extractors.

\begin{protocol}[H]
\caption{Secret key agreement protocol}
\label{prot:high_reliability}
\KwIn{Observations $X$ and $Y$} 
\KwOut{Secret key estimates $K_x$ and $K_y$}
\underline{Information reconciliation}
\\
Initiate the protocol with $l=1$
\\
\While{$l \leq L$ and ACK not received} 
{ 
First party sends the random bin index of $X$ into $M_l$ bins, 
$B_l = F_{1l}(X)$, to the second party
\\
\eIf{Second party find a {\it unique} $x$ such that $(x, Y) \in \cT_l$
and
%\[
$F_{1j}(x) = B_j, \quad \forall\,1\leq j \leq l$
%\]
}
{
Second party sets $\hat X =x$ and sends back
an ACK $F_{2i} = 1$ to the first party 
}
{
Second party sends back a NACK $F_{2i} = 0$\\
Parties update $l \rightarrow l+1$
}
}
\eIf{No ACK received}
{
Protocol declares an error and aborts
}
{
\underline{Privacy amplification}
\\
First party generates the random seed $S$ and
sends it to the second party using public communication
\\
First party generates the secret key $K_x = K = f_S(X)$ 
%\begin{align*}
%K =\begin{cases}
%f_S(X),\quad &I \in \cI_g,\\
%\mathtt{unif}(\cK),\quad &\text{otherwise}.
%\end{cases}
%\end{align*}
\\
The second party generates the estimate $K_y$ of $K$
as $K_y = f_S(\hat X)$
}
\end{protocol}
%%%%%%%%%

\begin{remark} \label{remark:leakage}
Note that since each ACK-NACK signal will require $1$-bit of communication to implement, 
the number of bits physically sent in the protocol is roughly $h(X|Y)  + \Delta + L$. However, the 
$\log$ of the number of values taken by the transcript is much less, roughly $h(X|Y)  + \Delta + \log L$, 
since the ACK-NACK sequence will be a stopped sequence consisting of NACKs followed by a single ACK. 
By Lemma~\ref{l:leftover_hash}, it is this latter quantity $h(X|Y)  + \Delta + \log L$ that captures the amount
of information leaked to the eavesdropper by public communication, which will be used
in our security analysis of the protocol.
\end{remark}
%%%%%%%%%%%%%

\subsection{Performance analysis of the protocol}
We now derive performance guarantees for the secret key agreement 
protocol of the previous section. In view of Proposition~\ref{p:time_sharing}
and Remark~\ref{r:actual_key}, the protocol above will constitute 
an $(\ep, \delta)$-SK protocol for arbitrary $\ep, \delta \in (0,1)$ if
it yields an $(\eta, \ep+\delta)$-SK with $\eta\approx 0$.
The result below shows that Protocol~\ref{prot:high_reliability}
indeed yields such a high reliability secret key.

Denote by $\idxy$ the information density
\begin{align} 
\idxy := \log \frac{\bP{XY}{x,y}}{\bP{X}{x}\bP{Y}{y}}.
\nonumber
\end{align}
\begin{theorem} \label{t:protocolA2-Z}
For $\lambda_{\min}, \lambda_{\max}, \Delta>0$ with $\lambda_{\max} \geq \lambda_{\min}$, let 
\begin{align*}
L = \frac{\lambda_{\max} - \lambda_{\min}}{\Delta}.
\end{align*}
Then, for every $\gamma> 0$ and $\lambda\ge 0$, there exists 
an $(\ep, \delta)$-SK $K$ taking values in $\cK$ with
\begin{align}
\ep &\leq \bP{XY}{\cT_0} + L 2^{-\gamma}, %\label{eq:single-shot-bound-ep-general-Z}
\nonumber
\\
\delta &\leq \bPr{ \idXY  -  \idXZ  \leq \lambda + \Delta} \nonumber \\
&+ \frac{1}{2} \sqrt{|\cK|2^{-(\lambda -\gamma - 3\log L)}}  + \frac{1}{L} + \bP{XY}{\cT_0} + L 2^{-\gamma},
% \label{eq:single-shot-bound-delta-general-Z}
\nonumber
\end{align}
where
\[
 \cT_0 := \left\{ (x,y) : -\log\bP{X|Y}{x|y} \geq \lambda_{\max} \text{ or } -\log\bP{X|Y}{x|y} < \lambda_{\min}\right\} 
\]
\end{theorem}
{\it Proof.} We begin by analyzing the reliability of Protocol~\ref{prot:high_reliability}.
Let $\rho(X,Y)$ denote the number of rounds after which the protocol stops when the observations are $(X,Y)$. 
An error occurs if $(X,Y)\in \cT_0$ or if there exists a $\hat{x}\neq X$ such that $(\hat x,Y)\in \cT_l$   
and $F_{1j}(X) = F_{1j}(\hat x)$ for all $j$ such that $1 \leq j \leq l$, for some $l \leq \rho(X,Y)$. 
Note that for each $1\leq j \leq L$, 
\[
|\{x : (x,y) \in \cT_j\}| \leq \exp(\lambda_j + \Delta) \quad \forall\, y\in \cY. 
\]
Therefore, using a slight modification of the usual
probability of error analysis for random binning, the probability of error for 
Protocol~\ref{prot:high_reliability} is bounded above as
\begin{align}
P_e &\leq \bP{XY}{\cT_0} + \sum_{x,y} \bP{XY}{x,y} \sum_{l=1}^{\rho(x,y)}\sum_{\hat x\neq x}\bPr{F_{1j}(x) = F_{1j}(\hat x),\, \forall\, 1\leq j \leq l}
\indicator\big((\hat x, y)\in \cT_l\big)
\nonumber
\\
&\leq \bP{XY}{\cT_0} + \sum_{x,y} \bP{XY}{x,y} \sum_{l=1}^{\rho(x,y)}\sum_{\hat x\neq x}\frac{1}{M_1...M_l}\indicator\big((\hat x, y)\in \cT_l\big)
\nonumber
\\
&\le \bP{XY}{\cT_0} + \sum_{x,y} \bP{XY}{x,y} \sum_{l=1}^{\rho(x,y)}2^{-\lambda_l -\Delta - \gamma}|\{\hat x : (\hat x, y)\in \cT_l\}|
\nonumber
\\
&\leq \bP{XY}{\cT_0} + L\, 2^{-\gamma},
\label{e:error_high_reliability}
\end{align}
where we have used the fact that $\log M_1...M_l = \lambda_1 + l \Delta +\gamma = \lambda_l +\Delta +\gamma$.

We now establish the secrecy of the protocol. Our proof entails establishing 
secrecy of the protocol conditioned on each realization $J=j$ of an appropriately
defined RV, which roughly corresponds to the slice index for $(X,Y)$.
Specifically, denote by $\cE_1$ the set of $(x,y)$
for which an error occurs in information reconciliation and 
by $\cE_2$ the set
\begin{align}
 \cE_2 := \left\{ (x,y,z) :
  \idxy  -  \idxz  \leq \lambda + \Delta  \right\}, \nonumber
\end{align}
which is the same as
\begin{align}
  \left\{ (x,y,z) : 
  \log\frac{1}{\bP{X|Z}{x|z}}  -  \log\frac{1}{\bP{X|Y}{x|y}}  \leq \lambda + \Delta  \right\}
\nonumber 
% \label{eq:definition-E2-with-Z}
\end{align}
Let RV $J$ taking values in the set $\{0,1,\ldots,L\}$ be defined as follows:
\[
J=
\begin{cases}
&0, \text{ if } (X,Y) \in \cT_0\cup \cE_1 \text{ or } (X,Y,Z) \in  \cE_2,
\\
&j \text{ if }(X,Y) \in \cT_j\cap \cE_1^c \text{ and } (X,Y,Z) \in \cE_2^c, 1\leq j \leq L.
\end{cases} 
\]
While we have used random coding in the information reconciliation stage, it is only for the ease of
proof and the encoder can be easily derandomized.\footnote{Since the final error probability is given by the expected probability of error,
where the expectation is over the source distribution and the additional shared randomness used in communication, there exists a realization of the shared randomness
for which the same expected error perfomance with respect to the source distribution is attained. } 
For the remainder of the proof, we assume a 
deterministic encoder; in particular, $J$ is a function of $(X,Y, Z)$.

We divide the indices $0 \leq j \leq L$ into good indices $\cI_g$
and the bad indices $\cI_b = \cI_g^c$, where
\begin{align} 
\cI_g = \left\{j : j>0 \text{ and }\bP{J}{j} \geq \frac{1}{L^2}\right\}.
%\label{eq:definition-of-good-indices}
\nonumber
\end{align}
Denoting by $\bF^l$ the communication up to $l$ round of the protocol,
\ie, $\bF^l := \{(F_{1j},F_{2j}), \, 1\leq j \leq l\}$ and by $\bF=\bF^{\rho(X,Y)}$
the overall communication, we have
\begin{align}
&\ttlvrn{ \bPP{K \bF Z S}}
{\bPP{\mathtt{unif}}\bPP{\bF Z S}} 
\nonumber
\\
&\leq \ttlvrn{\bPP{K \bF Z S J}}
{\bPP{\mathtt{unif}}\bPP{\bF Z S J }} 
\nonumber
\\
&\leq \bPr{J \in \cI_b} + \sum_{j\in \cI_g} \bP J j 
\ttlvrn{\bPP{K\bF Z S \mid J =j}}{ \bPP{\mathtt{unif}}\bPP{\bF Z S |J=j}}
\nonumber
\\
&\leq \bP{XY}{\cT_0 \cup \cE_1} + \bP{XYZ}{\cE_2} + \frac{1}{L} 
+\sum_{j\in \cI_g} \bP J j \ttlvrn{\bPP{K\bF Z S\mid J =j}}{ \bPP{\mathtt{unif}}\bPP{\bF Z S |J=j}}.
\label{eq:bound-on-secrecy-final-form}
\end{align}
To bound each term $\ttlvrn{\bPP{K\bF Z S \mid J =j}}{\bPP{\mathtt{unif}}\bPP{\bF Z S |J=j}}$, $j \in \cI_g$,
first note that under each event $J= j \in \cI_g$ information 
reconciliation succeeds and $\bF = \bF^j$.
Furthermore, the number of possible transcripts 
sent in the information reconciliation stage up to $j$th round, 
\ie, the cardinality $\|\bF^j\|$ of the range of RV $\bF^j$,
satisfies (cf.~Remark \ref{remark:leakage})
\[
\log\|\bF^j\| \leq \lambda_j + \Delta + \gamma + \log j.
\]
Let $\bPP{j}$ be the probability distribution of $X,Y,Z$ given $J= j$, \ie, 
\[
\bP{j}{x,y,z} := \frac{\bP{XYZ}{x,y,z} \indicator\left( J(x,y,z) = j\right)}{\bP J j}, \quad x\in \cX, y\in \cY, z\in \cZ,
0\leq j \leq L.
\]
With $\bPP{j,XZ}$ denoting the marginal on $\cX \times \cZ$ induced by $\bPP{j}$, for all $j \in \cI_g$, we have
\begin{align*}
\log \frac{\bP{j,XZ}{x,z}}{\bP{Z}{z}} 
&= \log \frac{\sum_y \bP{XYZ}{x,y,z}\indicator\left( J(x,y,z) = j\right)}{\bP Jj \bP{Z}{z}}  \\
&\le \log \frac{\sum_y 2^{-\lambda -\Delta} \bP{X|Y}{x|y} \bP{Y|XZ}{y|x,z}
\indicator\left( J(x,y,z) = j\right)}{\bP J j } \\
&\le \log \frac{\sum_y 2^{-\lambda_i - \lambda - \Delta} \bP{Y|XZ}{y|x,z}
\indicator\left( J(x,y,z) = j\right)}{\bP J j }
\\
&\le \log \frac{2^{-\lambda_i - \lambda - \Delta}}{\bP J j }
\\
&\le - \lambda_i - \lambda - \Delta + 2 \log L,
\end{align*}
where the first inequality holds since $J(x,y,z) > 0$ implies $(x,y,z) \in \cE_2^c$,
the second inequality holds since $J(x,y,z) = j$ implies $(x,y) \in \cT_j$, and the last inequality
holds since $j \in \cI_g$ implies $\bP{J}{j} > \frac{1}{L^2}$.
Thus, we obtain the following bound on $H_{\min}(\bPP{j, XZ}|\bPP{Z})$:
\[
H_{\min}(\bPP{j,XZ}|\bPP{Z}) \ge \lambda_j + \lambda + \Delta - 2 \log L.
\]
Therefore, noting that $S$ is independent of $(X, Z, \bF, J)$ 
and using Lemma~\ref{l:leftover_hash}, we get
\begin{align}
\ttlvrn {\bPP{K \bF Z S|J=j}}
{\bPP{\mathtt{unif}}\bPP{\bF Z S|J=j}} 
&=\ttlvrn {\bPP{K \bF^j Z S|J=j}}
{\bPP{\mathtt{unif}}\bPP{\bF^j Z S|J=j}} 
\nonumber
\\
&\leq \frac{1}{2} \sqrt{|\cK|\|\bF^j \|2^{-H_{\min}\left(\bPP{j,XZ} | \bPP{Z}\right)}}
 %\label{eq:bound-on-secrecy-for-given-I} 
\nonumber
\\
&\le \frac{1}{2} \sqrt{|\cK|2^{-(\lambda -\gamma - 3\log L)}}, \quad j\in \cI_g.
\nonumber
\end{align}
which gives the claimed secrecy by using the definition of $\cE_2$
and bounding $\bP{XY}{\cT_0 \cup \cE_1}$ using the union bound as in 
\eqref{e:error_high_reliability}.\qed

Thus, when the secret key length $\log |\cK| \approx \lambda$,
the reliability parameter $\ep$ for Protocol~\ref{prot:high_reliability}
can be made very small by appropriately choosing parameters $\la_{\min}$ and $\la_{\max}$,
and the secrecy parameter $\delta$ can be made roughly as small as the
tail-probability $ \bPr{ \idXY  -  \idXZ  \leq \lambda}$ . In fact, Proposition~\ref{p:time_sharing}
allows us to shift this constraint on $\delta$ to a constraint on $\ep+\delta$
and Protocol~\ref{prot:high_reliability} yields
an $(\ep, \delta)$-SK of length roughly equal to $\lambda$ 
as long as $\ep+\delta$ is greater than   $ \bPr{ \idXY  -  \idXZ  \leq \lambda}$.
Formally, we have the following simple corollary of Theorem~\ref{t:protocolA2-Z}.
 %%% Corollary %%%%
\begin{corollary} \label{corollary-time-sharing}
For $\lambda_{\max},\lambda_{\min},\Delta>0$ with $\lambda_{\max} \ge \lambda_{\min}$, let
\begin{align*}
L = \frac{\lambda_{\max} - \lambda_{\min}}{\Delta},
\end{align*}
and let
\[
 \cT_0 := \left\{ (x,y) : -\log\bP{X|Y}{x|y} \geq \lambda_{\max} \text{ or } -\log\bP{X|Y}{x|y} < \lambda_{\min}\right\}. 
\]
Then, for every $\lambda\ge 0$ and every $\ep$ and $\delta$ satisfying 
\begin{align}
\ep &\ge \bP{XY}{\cT_0} + \frac{1}{4} \left(L^4|\cK|2^{-\lambda}\right)^{\frac 13}, 
\label{eq:condition-epsilon} 
\\
\ep + \delta &\ge \bPr{ \idXY  -  \idXZ  \leq \lambda + \Delta}
 + \frac{1}{L} +2\bP{XY}{\cT_0}
+ \frac{3}{2} \left(L^4|\cK|2^{-\lambda}\right)^{\frac 13}, 
\label{eq:condition-sum-epsilon-delta}
\end{align}
there exists an $(\ep,\delta)$-SK $K$ taking values in $\cK$.
\end{corollary}
\begin{proof}
Let 
\begin{align*}
\eta(\gamma) &:= \bP{XY}{\cT_0} + L 2^{-\gamma}, \\
\alpha(\gamma) &:=  \bPr{ \idXY  -  \idXZ  \leq \lambda + \Delta}
+ \frac{1}{2} \sqrt{|\cK|2^{-(\lambda -\gamma - 3\log L)}}  + \frac{1}{L} + \bP{XY}{\cT_0} + L 2^{-\gamma}.
\end{align*}
We first optimize 
\begin{align*}
&\eta(\gamma) + \alpha(\gamma) \\
&= \bPr{ \idXY  -  \idXZ  \leq \lambda + \Delta}
+ \frac{1}{2} \sqrt{|\cK|2^{-(\lambda -\gamma - 3\log L)}} 
+ \frac{1}{L} + 2\bP{XY}{\cT_0} +2L 2^{-\gamma} 
\end{align*}
over $\gamma$. By setting $a=L2^{-\gamma}$ and by noting that the 
function $f(a) = 2a + \frac{a^{-1/2}L^2}{2} \sqrt{|\cK|2^{-\lambda}}$
has minimum value $\frac{3}{2}\left(L^4|\cK|2^{-\lambda}\right)^{\frac 13}$ with 
$a =\frac{1}{4} \left(L^4|\cK| 2^{-\lambda}\right)^{\frac{1}{3}}$, the minimum of 
$\eta(\gamma) + \alpha(\gamma)$ is achieved when 
\begin{align*}
\eta = \eta^* &:= \bP{XY}{\cT_0} + \frac{1}{4} \left(L^4|\cK| 2^{-\lambda}\right)^{\frac{1}{3}}, \\
\alpha = \alpha^* &:= \bPr{ \idXY  -  \idXZ  \leq \lambda + \Delta}
 + \frac{1}{L} + \bP{XY}{\cT_0} + \frac{5}{4} \left(L^4|\cK| 2^{-\lambda}\right)^{\frac{1}{3}}.
\end{align*}
The corollary follows by applying Proposition~\ref{p:time_sharing} 
with $\eta=\eta^*$ and $\alpha^*$ to the resulting $(\eta^*, \alpha^*)$-SK.
\end{proof}

%%%%%%%%%%%%%%%%%%%%%%%%%%%%%%%%%%%%%%%%%%%%%%%
\section{Secret key capacity for general sources}\label{s:general_sources}
In this section we will establish the secret key capacity for a sequence of
general sources $(X_n, Y_n, Z_n)$ with joint distribution\footnote{The distributions $\bPP{X_nY_nZ_n}$ 
need not satisfy the consistency conditions.} $\bPP{X_nY_nZ_n}$.
The secret key capacity for general sources is defined as follows \cite{Mau93, AhlCsi93, CsiNar04}.
\begin{definition}
The secret key capacity $C$ is defined as
\begin{align*}
C := \sup_{\ep_n, \delta_n}\liminf_{n\to\infty} \frac{1}{n} S_{\ep_n, \delta_n}\left(X_n, Y_n\mid Z_n\right),
\end{align*}
where the $\sup$ is over all $\ep_n,\delta_n \geq 0$ such that 
$$\lim_{n\rightarrow \infty}\ep_n+\delta_n = 0.$$
\end{definition}
To state our result, we need the following concepts from the {\it information spectrum}
method; see \cite{Han03} for a detailed account. For RVs $(X_n, Y_n, Z_n)_{n=1}^\infty$, the
{\it inf-conditional entropy rate} $\underline{H}(\bX\mid \bY)$ and the {\it sup-conditional entropy rate} $\overline{H}(\bX\mid \bY)$
are defined as follows:
\begin{align*}
\underline{H}(\bX\mid \bY) &= \sup\left\{\alpha \mid \lim_{n\rightarrow \infty} \bPr{-\frac{1}{n}\log\bP{X_n\mid Y_n}{X_n\mid Y_n} < \alpha} = 0\right\},
\\
\overline{H}(\bX\mid \bY) &= \inf\left\{\alpha \mid \lim_{n\rightarrow \infty} \bPr{-\frac{1}{n}\log\bP{X_n\mid Y_n}{X_n\mid Y_n} > \alpha} = 0\right\}.
\end{align*}
Similarly, the {\it inf-conditional information rate} $\underline{I}({\bf X} \wedge {\bf Y}\mid \bZ)$ is defined as
\begin{align*}
\underline{I}({\bf X} \wedge {\bf Y}\mid \bZ) &= \sup\bigg\{\alpha \mid \quad \lim_{n\rightarrow \infty} \bPr{\frac 1n\, i(X_n, Y_n\mid Z_n)< \alpha} = 0\bigg\},
\end{align*}
where, with a slight abuse of notation, $i(X_n, Y_n\mid Z_n)$ denotes the conditional information density
$$i(X_n, Y_n\mid Z_n) = \log \frac{\bP{X_nY_n\mid Z_n}{X_n,Y_n\mid Z_n}}{\bP{X_n\mid Z_n}{X_n\mid Z_n} \bP{Y_n\mid Z_n}{Y_n\mid Z_n}}.$$
We also need the following result credited to Verd\'u.
\begin{lemma}{\bf \cite[Theorem 4.1.1]{Han03}}\label{l:bound_on_beta} For every $\ep_n$ such that
$$\lim_{n \rightarrow \infty} \ep_n =0,$$ 
it holds that
$$\liminf_{n} -\frac{1}{n}\log\beta_{\ep_n}\left(\bPP{X_nY_nZ_n}, \bPP{X_n\mid Z_n}\bPP{Y_n\mid Z_n}\bPP{Z_n}\right) \leq \underline{I}(\bX \wedge \bY\mid \bZ),$$
where $\beta_{\ep_n}$ is defined in \eqref{e:beta-epsilon}.
\end{lemma}
Our result below characterizes the secret key capacity $C$ for general sources
for the special case when $X_n \mc Y_n \mc Z_n$ is a Markov chain. 
%%%%%%
\begin{theorem}\label{t:secret key_capacity}
For a sequence of sources $\{X_n, Y_n, Z_n\}_{n=1}^\infty$ 
such that $X_n \mc Y_n \mc Z_n$ form a Markov chain for all $n$,
the secret key capacity $C$ is given by\footnote{We assume that $\overline{H}({\bf X}\mid \bY) < \infty$.} 
$$C= \underline{I}(\bX \wedge \bY\mid \bZ).$$
\end{theorem}
%%%%%%
{\it Proof.} Applying Theorem \ref{theorem:one-shot-converse-source-model} with $\eta = \eta_n = n^{-1}$,
along with Lemma \ref{l:bound_on_beta}, gives 
$$C\leq \underline{I}({\bf X} \wedge {\bf Y}\mid \bZ).$$
For the other direction, we construct a sequence of $(\ep_n,\delta_n)$-SKs $K = K_n$
with $\ep_n, \delta_n\rightarrow 0$ and rate approximately $\underline{I}({\bf X} \wedge {\bf Y}\mid \bZ)$.
Indeed,  in Theorem \ref{t:protocolA2-Z} choose
\begin{align*}
\lambda_{\max} &= n\left(\overline{H}({\bf X}\mid \bY) + \Delta\right),
\\
\lambda_{\min} &= n\left(\underline{H}({\bf X}\mid \bY) - \Delta\right),
\\
\gamma &= \gamma_n = n\Delta/2,
\\
\lambda &= \lambda_n =  n\left(\underline{I}\left({\bf X} \wedge {\bf Y}\mid \bZ \right) - \Delta\right);
\end{align*}
thus, 
$$L = L_n = \frac{n\left(\overline{H}({\bf X}\mid \bY) - \underline{H}({\bf X}\mid \bY) + 2\Delta\right)}{\Delta}.$$
Since $\idXY -\idXZ  = i(X,Y|Z)$ if $X\mc Y\mc Z$ form a Markov chain,
there exists an $(\ep_n, \delta_n)$-SK $K_n$ of rate given by
\begin{align*}
\frac{1}{n}\log|\cK| &= \frac{1}{n}\left(\lambda_n - 3\log L_n\right) - \Delta
\\
&= \underline{I}\left({\bf X} \wedge {\bf Y}\mid \bZ\right) - 2\Delta  - o(n),
\end{align*}
such that $\ep_n,\delta_n \rightarrow 0$ as $n\rightarrow \infty$.
Rates arbitrarily close to $\underline{I}({\bf X} \wedge {\bf Y}\mid \bZ)$ 
are achieved by this scheme as $\Delta>0$ is arbitrary.
\qed

In the achievability part of the proof above, we actually show that, in general, 
our protocol generates a secret key of rate
\[
\sup\left\{\alpha \mid \lim_{n\rightarrow \infty} \bPr{\frac{1}{n}[i(X_n, Y_n) - i(X_n, Z_n)] < \alpha} = 0 \right\},
\]
which matches the converse bound of $\underline{I}\left({\bf X} \wedge {\bf Y}\mid \bZ\right)$
in the special case when $X_n \mc Y_n \mc Z_n$ holds.
%%%%%%%%%%%%%%%%%%%%%%%%%%%%%%%%%%%%%%%%%%%%%%%
\section{Second-order asymptotics of secret key rates}\label{s:second_order}
The results of the previous section show that for
with $\ep_n, \delta_n\rightarrow 0$,
the largest length $S_{\ep_n, \delta_n}(X_n, Y_n \mid Z_n)$ of an $(\ep_n,\delta_n)$-SK $K$ is 
\begin{align}
\sup_{\ep_n, \delta_n} S_{\ep_n, \delta_n}(X_n, Y_n \mid Z_n)= n \underline{I}(\bX \wedge \bY \mid \bZ) + o(n),
\label{e:order_n_term}
\end{align}
if $X_n \mc Y_n \mc Z_n$ form a Markov chain. For the case when $(X_n, Y_n, Z_n) = (X^n, Y^n, Z^n)$ is the $n$-IID repetition
of $(X, Y, Z)$ where $X\mc Y \mc Z$, we have 
$$\underline{I}(\bX \wedge \bY \mid \bZ)  = I(X \wedge Y \mid Z).$$ 
Furthermore, \eqref{e:order_n_term}
holds even without $\ep_n, \delta_n \rightarrow 0$. In fact, a finer asymptotic analysis 
is possible and the second-order asymptotic term in the maximum length of an $(\ep, \delta)$-SK
can be established; this is the subject-matter of the current section.

Let 
\begin{align*}
V := \var\left[ i(X,Y \mid Z) \right],
\end{align*}
and let
%\begin{align}
%i(x,y \mid z) := \log \frac{\bP{XY \mid Z}{x,y \mid z}}{\bP{X \mid Z}{x \mid z} \bP{Y \mid Z}{y \mid z}}.
%\end{align}
%Let
\begin{align*}
Q(a) := \int_{a}^\infty \frac{1}{\sqrt{2 \pi}} \exp\left[ - \frac{t^2}{2} \right] dt
\end{align*}
be the tail probability of the standard Gaussian distribution. 
Under the assumptions 
\begin{align}
V_{X|Y} &:= \var[-\log \bP{X|Y}{X|Y}] < \infty, \label{eq:variance-conditional-log-likelihood}
\\
T &:= \mathbb{E}\left[ \left| i(X,Y \mid Y) - I(X \wedge Y \mid Z) \right|^3 \right] < \infty,
\end{align}
the result below establishes the second-order asymptotic term in $S_{\ep, \delta}(X^n, Y^n \mid Z^n)$.

%%%%%%%%%%%
\begin{theorem}\label{t:second_order}
For every $\ep, \delta > 0$ such that $\ep + \delta < 1$ and IID RVs $(X^n, Y^n,Z^n)$
such that $X \mc Y \mc Z$ is a Markov chain, we have
\begin{align*}
S_{\ep, \delta}\left(X^n, Y^n \mid Z^n \right) = n I(X \wedge Y \mid Z) - \sqrt{n V} Q^{-1}(\ep+\delta) 
\pm  {\cal O}(\log n),
\end{align*}
\end{theorem}
{\it Proof.}
For the converse part, we proceed along the lines of \cite[Lemma 58]{PolPooVer10}.
Recall the following simple bound for $\beta_\ep(\dP, \dQ)$ (\cf~\cite[Lemma 4.1.2]{Han03}):
\[
-\log \beta_\ep(\dP, \dQ) \leq \lambda 
- \log \left( \bP {}{\left\{ x: \log \frac{\dP(x)}{\dQ(x)} \leq \la\right\}} 
- \ep\right).
\]
Thus, applying Theorem \ref{theorem:one-shot-converse-source-model} with
$\bPP{XYZ} = \bPP{X^nY^nZ^n}$, $\mathrm{Q}_{XYZ} = \bPP{X^n|Z^n} \bPP{Y^nZ^n}$, and
$\eta = \eta_n = n^{-1/2}$, and choosing
\[
\la = nI(X\wedge Y \mid Z) - \sqrt{nV}Q^{-1}\left(\ep + \delta +\theta_n\right),
\]
where
\[
\theta_n = \frac 2 {\sqrt{n}} + \frac {T^3} {2V^{3/2}\sqrt{n}},
\]
we get by the Berry-Ess\'een theorem (\cf~\cite{Fel71, Shevtsova11}) that
\begin{align*}
\bP{}{ i(X^n,Y^n) - i(X^n, Z^n) \leq \la }  \ge \ep + \delta + \frac{2}{\sqrt{n}},
\end{align*}
which implies
\begin{align}
S_{\ep, \delta}(X^n, Y^n\mid Z^n) &\leq 
\lambda 
- \log \left( \bP {}{i(X^n,Y^n) - i(X^n, Z^n) \leq \la} 
- \ep -\delta - n^{-1/2} \right) + \log n
\nonumber
\\
&\leq nI(X\wedge Y \mid Z) - \sqrt{nV}Q^{-1}\left(\ep + \delta + \theta_n\right)
+ \frac{3}{2} \log n.
\label{e:nonasymptotic_upper_bound}
\end{align}
Thus, we  have the desired converse by using Taylor approximation of $Q(\cdot)$ to remove
$\theta_n$.

For the direct part, we use Corollary~\ref{corollary-time-sharing} by setting
\begin{align*}
\lambda_{\max} &= n (H(X|Y) + \Delta/2), \\
\lambda_{\min} &= n (H(X|Y) - \Delta/2), \\
\lambda &= n I(X \wedge Y \mid Z) - \sqrt{nV} Q^{-1}(\ep+\delta - \theta_n') - \Delta,
\end{align*}
and 
\begin{align}
\log|\cK| & = n I(X \wedge Y \mid Z) - \sqrt{nV} Q^{-1}(\ep+\delta - \theta_n') - \frac{11}{2}\log n -\Delta,
\label{e:nonasymptotic_lower_bound}
\end{align}
where $0<\Delta < 2H(X\mid Y)$, and 
\begin{align*}
\theta_n' = \frac{8 V_{X\mid Y}}{n \Delta^2} + \frac{T^3}{2V^{3/2} \sqrt{n}} + \frac{1}{n} + \frac{3}{2\sqrt{n}}.
\end{align*}
Note that $L=n$. Upon bounding the term $\bP{XY}{\cT_0}$ in \eqref{eq:condition-epsilon} and
\eqref{eq:condition-sum-epsilon-delta} by $\frac{4 V_{X|Y}}{n\Delta^2}$ using Chebyshev's inequality,
the condition \eqref{eq:condition-epsilon} is satisfied for sufficiently large $n$. Furthermore, upon
bounding the first term of \eqref{eq:condition-sum-epsilon-delta} by the Berry-Ess\'een theorem,
the condition \eqref{eq:condition-sum-epsilon-delta} is also satisfied. 
Thus, it follows that there exists an $(\ep,\delta)$-SK $K$ taking values on $\cK$.
The direct part follows by using Taylor approximation of $Q(\cdot)$ to remove $\theta_n'$. 
\qed
 
%\begin{remark}\label{r:second_order_restricted}
%The second-order asymptotic expansion of Theorem~\ref{t:second_order}
%does not hold for the more demanding secrecy requirement \eqref{eq:secrecy-condition-demanding},
%since in the absence of Proposition~\ref{p:conversion} a high reliability secret key construction
%alone does not suffice. However, when $Z$ is a constant, 
%we can use a convex combination of the high secrecy and the high reliability protocol of
%Section~\ref{subsection:Z-constant} to construct $(\ep, \delta)$-SKs even under \eqref{eq:secrecy-condition-demanding}.
%Thus, when $Z$ is a constant, the analysis above extends even for secrecy requirement \eqref{eq:secrecy-condition-demanding}
%and the second-order expansion in Theorem~\ref{t:second_order} holds under \eqref{eq:secrecy-condition-demanding}.
%\end{remark} 

\begin{remark}\label{r:noninteractive_insufficient}
Note that  a standard noninteractive secret key agreement protocol based on information 
reconciliation and privacy amplification (cf.~\cite{RenWol05}) only gives the following 
suboptimal achievability bound on the second-order asymptotic term:
\begin{align*}
S_{\ep, \delta}\left(X^n, Y^n \mid Z^n \right) \ge n I(X \wedge Y \mid Z) - \sqrt{n V_{X|Y}} Q^{-1}(\ep) - \sqrt{n V_{X|Z}} Q^{-1}(\delta) 
+  o(\sqrt{n}),
\end{align*}
where $V_{X|Y}$ and 
$V_{X|Z}$ are the variances of the conditional 
log-likelihoods of $X$ given $Y$ and $Z$ respectively (cf.~\eqref{eq:variance-conditional-log-likelihood}).
\end{remark}

We close this section with a numerical example that illustrates the utility of our bounds
in characterizing the gap to secret key capacity at a fixed $n$.
\begin{example}[{{\bf Gap to secret key capacity}}]\label{e:gap_to_capacity}
For $\alpha_0, \alpha_1 \in (0,1/2)$, 
let $B_0$ and $B_1$ be independent random bits taking value $1$
with probability $\alpha_0$ and $\alpha_1$, respectively.
Consider binary $X, Y, Z$ where $Z$ is a uniform random bit independent
jointly of $B_0$ and $B_1$, $Y = Z\oplus B_0$, and $X = Y \oplus B_1$. We consider the rate 
$S_{\ep, \delta}(X^n, Y^n \mid Z^n)/ n$ of an $(\ep, \delta)$-SK that can be generated using
$n$ IID copies of $X$ and $Y$ when the eavesdropper observes $Z^n$. The following 
quantities, needed to evaluate \eqref{e:nonasymptotic_upper_bound}  and 
\eqref{e:nonasymptotic_lower_bound}, can be easily evaluated:
\begin{align*}
I(X\wedge Y \mid Z) &= h(\alpha_0*\alpha_1) - h(\alpha_1),
\quad
V = \mu_2,
\quad
T = \mu_3,
\\
V_{X|Y} &= \alpha_1 (\log \alpha_1 - h(\alpha_1))^2 + (1- \alpha_1)(\log (1-\alpha_1) - h(\alpha_1))^2,
\end{align*}
where $\mu_r$ is the $r$th central moment of $i(X,Y \mid Z)$ and is given by
\begin{align*}
\mu_r &=  
\alpha_0\alpha_1\left|\log \frac{\alpha_1}{1 - \alpha_0*\alpha_1} - I(X\wedge Y \mid Z)\right|^r
+ (1-\alpha_0)(1-\alpha_1)
\left|\log \frac{1-\alpha_1}{1 - \alpha_0*\alpha_1} -I(X\wedge Y \mid Z)\right|^r
\\
&\quad + 
(1-\alpha_0)\alpha_1  \left|\log \frac{\alpha_1}{ \alpha_0*\alpha_1} - I(X\wedge Y \mid Z)\right|^r
+ \alpha_0(1-\alpha_1) 
\left|\log \frac{1 - \alpha_1}{ \alpha_0*\alpha_1} - I(X\wedge Y \mid Z)\right|^r,
\end{align*}
$h(x) = - x\log x  - (1-x)\log (1-x)$ is the binary entropy function and 
$\alpha_0*\alpha_1 = \alpha_0(1- \alpha_1)+ (1-\alpha_0)\alpha_1$. 
In Figure~\ref{f:gap_to_capacity} (given in Section~\ref{s:introduction}), we 
plot the upper bound on $S_{\ep, \delta}(X^n, Y^n \mid Z^n)/ n$ resulting 
from~\eqref{e:nonasymptotic_upper_bound} and the lower bound resulting 
from~\eqref{e:nonasymptotic_upper_bound} with $\Delta = 1$ for 
$\alpha_0 =0.25$ and $\alpha_1 =0.125$.
\end{example}

%%%%%%%%%%%%%%%%%%%%%%%%%%%%%%%%%%%%%%%%%%%%%%
\section{Discussion: Is interaction necessary?}\label{s:discussion}
In contrast to the protocols in \cite{Mau93, AhlCsi93, CsiNar04, RenWol05}, our proposed Protocol~\ref{prot:high_reliability}
for secret key agreement is interactive. In fact, the protocol requires as many rounds of 
interaction as the number of slices $L$, which can be pretty large in general.
For instance, to obtain the second-order asymptotic term in the previous
section, we chose $L =n$. In Appendix~\ref{appendix:high_secrecy_protocol}, we present an
alternative protocol which requires only $1$-bit of feedback and, in the special case when $Z$ is constant, 
achieves the asymptotic results of Sections~\ref{s:general_sources} and~\ref{s:second_order}.
But is interaction necessary for attaining our asymptotic results? Below we present an example 
where none of the known (noninteractive) secret key agreement 
protocols achieves the general capacity of Theorem~\ref{t:secret key_capacity}, 
suggesting that perhaps interaction is necessary. 

For $i=1,2$, let $(X_i^n,Y_i^n,Z_i^n)$ be IID with $\cX = \cY = \cZ = \{0,1\}$ such that 
\begin{align*}
\bP{X_i^nY_i^nZ_i^n}{x^n, y^n,z^n} = \frac{1}{2^n} W_i^n(y^n|x^n) V^n(z^n|y^n),
\end{align*}
where $W_i$ and $V$, respectively, are binary symmetric channels with crossover probabilities 
$p_i$ and $q$. Let $(X_n,Y_n,Z_n)$ be the mixed source given by
\begin{align*}
\bP{X_n Y_n Z_n}{x^n,y^n,z^n} &= \frac{1}{2} \bP{X_1^nY_1^nZ_1^n}{x^n,y^n,z^n} 
 + \frac{1}{2} \bP{X_2^nY_2^nZ_2^n}{x^n,y^n,z^n} \\
&= \frac{1}{2^n} \left[ \frac{1}{2} W_1^n(y^n|x^n) + \frac{1}{2} W_2^n(y^n|x^n) \right] V^n(z^n|y^n).
\end{align*} 
Note that $X_n \mc Y_n \mc Z_n$ forms a Markov chain. Suppose that $0<p_1 < p_2<\frac{1}{2}$. Then, we have
\begin{align*}
\underline{I}(\bX \wedge \bY\mid \bZ) 
&= \min[H(X_1|Z_1) - H(X_1|Y_1), H(X_2|Z_2) - H(X_2|Y_2)] \\
&= \min[ h(p_1 * q) - h(p_1), h(p_2 * q) - h(p_2)] \\
&= h(p_2*q) - h(p_2),
\end{align*}
where $h(\cdot)$ is the binary entropy function and $*$ is binary convolution. 
Using a standard noninteractive secret key agreement protocol based on information 
reconciliation and privacy amplification (cf.~\cite{RenWol05}), we can achieve
only
\begin{align*}
& \underline{H}(\bX\mid \bZ) - \overline{H}(\bX \mid \bY) \\
&= \min[H(X_1|Z_1), H(X_2|Z_2)] - \max[H(X_1|Y_1), H(X_2|Y_2)] \\
&= H(X_1|Z_1) - H(X_2|Y_2) \\
&= h(p_1 * q) - h(p_2),
\end{align*}
which is less than the general secret key capacity of 
Theorem~\ref{t:secret key_capacity}. Proving a precise 
limitation result for noninteractive protocols is a 
direction for future research.
%%%%%%%%%%%%%%%%%%%%%%%%%%%%%%%%%%%%%%%%%%%%%%%
\appendix

\subsection{Proof of Theorem \ref{theorem:one-shot-converse-source-model}}
\label{appendix:theorem:one-shot-converse-source-model}

The definition of a secret key used in \cite{TyaWat14, TyaWat14ii}
is different from the one in Definition~\ref{d:secret_key}, and it conveniently combines the secrecy
and the reliability requirements into a single expression. Instead of
considering a separate RV $K$, the alternative definition
directly works with the estimates $K_x$ and $K_y$. Specifically,
let $K_x$ and $K_y$ be functions of $(U_x, X, \bF)$ and 
$(U_y, Y, \bF)$, respectively, where $\bF$ is an interactive communication.
Then, RVs $K_x$ and $K_y$ with a common range $\cK$ constitute an $\ep${\em-secret key}
($\ep$-SK) if 
\begin{eqnarray}
\ttlvrn{\bPP{K_xK_y\bF Z}}{\mathrm{P}_{\mathtt{unif}}^\two \times \bPP{\bF Z}} &\leq \ep,
\label{e:security}
\end{eqnarray}
where, for a pmf $\bPP {}$ on $\cX$, $\bPP {}^\m$ denotes its extension to
$\cX^m$ given by  
\[
\bPP {}^\m(x_1, ..., x_m) = \bP{}{x}\indicator(x_1 = ... =x_m), \quad (x_1, ..., x_m)\in \cX^m.
\]
Note that the alternative definition captures reliability condition $\bPr{K_x = K_y} \geq 1-\ep$
by requiring that the joint distribution $\bPP{K_xK_y}$ is close to a uniform distribution
on the diagonal of $\cK \times \cK$. The upper bound in \cite{TyaWat14, TyaWat14ii} holds
under this alternative definition of a secret key. However, the next lemma says that
this alternative definition is closely related to our Definition~\ref{d:secret_key}.

\begin{lemma}\label{l:secret key_combined}
Given $\ep, \delta \in [0,1)$ and an $(\ep, \delta)$-SK $K$, the local
estimates $K_x$ and $K_y$ satisfy \eqref{e:security} with
$\ep+\delta$, \ie, 
\begin{align}
\ttlvrn {\bPP {K_x K_y\bF Z}}{\mathrm{P}_\mathtt{unif}^\two \times \bPP {\bF Z}} \le \ep +\delta.
%\label{e:secret key_combined}
\nonumber
\end{align}
Conversely, if $K_x$ and $K_y$ satisfy \eqref{e:security},
either $K_x$ or $K_y$ constitutes an $(\ep, \ep)$-SK.
\end{lemma}
{\it Proof.} We prove the direct part first. For an $(\ep, \delta)$-SK $K$, 
\begin{align*}
&\espc\ttlvrn {\bPP {K_x K_y\bF Z}}{\mathrm{P}_\mathtt{unif}^\two \times \bPP {\bF Z}}
\\
&\leq \ttlvrn {\bPP {KK_x K_y\bF Z}}{\mathrm{P}_\mathtt{unif}^\three \times \bPP {\bF Z}}
\\
& \leq \ttlvrn {\bPP {KK_x K_y\bF Z}} {\mathrm{P}_{K | \bF Z}^\three\times\bPP{\bF Z}}
 + \ttlvrn  {\mathrm{P}_{K| \bF Z}^\three\times \bPP{\bF Z}}{\mathrm{P}_\mathtt{unif}^\three \times \bPP {\bF Z}}.
\end{align*}
Since 
\[
\ttlvrn \dP \dQ = \dQ(\{x: \dQ(x) \ge \dP(x)\}) - \dP(\{x: \dQ(x) \ge \dP(x)\})  
\]
and 
\begin{align}
& \{ (k,k_x,k_y,f,z) : \mathrm{P}_{K|\bF Z}^\three(k,k_x,k_y|f,z) 
 \ge \mathrm{P}_{KK_xK_y|\bF Z}(k,k_x,k_y|f,z) \}  
\nonumber
\\
 &= \{ (k,k_x,k_y,\mathbf{f},z) : k = k_x = k_y \}, 
\nonumber
%\label{e:variational_distance_indicator}
\end{align}
the first term on the right-side above satisfies
\begin{align}
 \ttlvrn {\bPP {KK_x K_y\bF Z}} {\mathrm{P}_{K | \bF Z}^\three\times \bPP{\bF Z}} = 1 - \bPr {K=K_x =K_y}
\le \ep.
\label{e:secret key_tv_to_prob} 
\end{align}
Furthermore, the second term satisfies
\begin{align*}
&\ttlvrn  {\mathrm{P}_{K| \bF Z}^\three\times \bPP{\bF Z}}{\mathrm{P}_\mathtt{unif}^\three \times \bPP {\bF Z}}
\\
&= \sum_{k, k_x, k_y, f, z}\bP{\bF Z}{f, z}
\left|\bP{K| \bF Z}{k|f,z}\indicator(k=k_x=k_y) 
- \indicator(k=k_x= k_y)\frac 1{|\cK|}\right|
\\
&=\ttlvrn{\bPP {K| \bF Z} \times \bPP{\bF Z}}{\mathrm{P}_\mathtt{unif} \times \bPP {\bF Z}}
\\
&\le \delta,
\end{align*}
where the last inequality is by the $\delta$-secrecy condition $K$.
Combining the bounds on the two terms above, the direct part follows.

For the converse, $\ep$-secrecy of $K_x$ (or $K_y$) holds since by the monotonicity of the variational distance
\[
\ttlvrn{\bPP{K_x\bF Z}}{\bPP{\mathtt{unif}}\times \bPP{\bF Z}}
\leq
\ttlvrn{\bPP{K_xK_y\bF Z}}{\mathrm{P}_{\mathtt{unif}}^\two \times \bPP{\bF Z}} \leq \ep.
\]
The $\ep$-reliability condition, too, follows from the triangle inequality upon
observing that
\begin{align*}
\ttlvrn{\bPP{K_xK_y}}{\bPP{\mathtt{unif}}^\two}
&= \sum_{k_x, k_y}\left|\bP {K_xK_y}{k_x, k_y}
- \indicator(k_x = k_y)\frac 1{|\cK|} 
\right|
\\
&\geq \sum_{k_x \neq k_y}\bP {K_xK_y}{k_x, k_y}
\\
&= \bPr{K_x\neq K_y}.
\end{align*}
\qed

To prove Theorem \ref{theorem:one-shot-converse-source-model}, we first relate 
the length of a secret key satisfying \eqref{e:security}
to the exponent of the probability of error of type II in a binary hypothesis 
testing problem where an observer of $(K_x,K_y,\bF,Z)$ seeks to find out
if the underlying distribution was $\bPP{X Y Z}$ of $\bQQ{XYZ} = \bQQ{X|Z} \bQQ{Y|Z} \bQQ{Z}$.
This result is stated next.

\begin{lemma} \label{lemma:relation-SK-HT}
For an $\ep$-SK $(K_x,K_y)$ satisfying \eqref{e:security}
generated by an interactive communication $\bF$, let
$W_{K_x K_y \bF|X Y Z}$ denote the resulting conditional distribution 
on $(K_x,K_y,\bF)$ given $(X,Y,Z)$.
Then, for every $0 < \eta < 1 - \ep$ and every $\bQQ{XYZ} = \bQQ{X|Z} \bQQ{Y|Z} \bQQ{Z}$, we have
\begin{align} \label{eq:relation-SK-HT}
\log |\cK | \le - \log \beta_{\ep+\eta}(\bPP{K_xK_y\bF Z}, \bQQ{K_xK_y\bF Z}) + 2 \log(1/\eta),
\end{align}
where $\bQQ{K_x K_y \bF Z}$ is the marginal of $(K_x,K_y,\bF,Z)$ of the joint distribution
\begin{align*}
\bQQ{K_xK_y \bF X Y Z} = \bQQ{XYZ} W_{K_x K_y \bF|XYZ}.
\end{align*}
\end{lemma}

To prove Lemma \ref{lemma:relation-SK-HT}, we need the following basic property of interactive communication ($cf.$~\cite{TyaNar13ii}).
\begin{lemma}[Interactive communication property] \label{lemma:property-interactive-communication}
Given $\bQQ{XYZ} = \bQQ{X|Z} \bQQ{Y|Z} \bQQ{Z}$ and an interactive communication $\bF$, the following holds:
\begin{align*}
\bQQ{X Y|\bF Z} = \bQQ{X|\bF Z} \times \bQQ{Y| \bF Z},
\end{align*}
i.e., conditionally independent observations remain so when conditioned additionally on an interactive communication.
\end{lemma}

\paragraph*{Proof of Lemma \ref{lemma:relation-SK-HT}}
We establish \eqref{eq:relation-SK-HT} by constructing a test for the hypothesis testing problem with
null hypothesis $\bPP{} = \bPP{K_x K_y \bF Z}$ and alternative hypothesis $\bQQ{} = \bQQ{K_x K_y \bF Z}$.
Specifically, we use a deterministic test with the following acceptance region (for the null hypothesis)\footnote{The
values $(k_x,k_y,f,z)$ with $\bQ{K_x K_y|\bF Z}{k_x,k_y|f,z} = 0$ are included in $\cA$.}:
\begin{align*}
\cA := \left\{ (k_x,k_y,f,z) : \log \frac{\mathrm{P}_{\mathtt{unif}}^\two(k_x,k_y)}{\bQ{K_x K_y|\bF Z}{k_x,k_y|f,z}} \ge \lambda \right\},
\end{align*}
where 
\begin{align*}
\lambda = \log |\cK | - 2 \log(1/\eta).
\end{align*}
For this test, the probability of type II is bounded above as 
\begin{align}
\bQ{K_x K_y \bF Z}{\cA} 
 &= \sum_{f,z} \bQ{\bF Z}{f,z} \sum_{k_x, k_y : \atop (k_x,k_y,f,z) \in \cA} \bQ{K_xK_y|\bF Z}{k_x,k_y|f,z} \nonumber \\
 &\le 2^{-\lambda}  \sum_{f,z} \bQ{\bF Z}{f,z} \sum_{k_x,k_y} \mathrm{P}_{\mathtt{unif}}^\two(k_x,k_y) \nonumber \\
 &= \frac{1}{|\cK| \eta^2}. \label{eq:type-two-bound}
\end{align}
On the other hand, the probability of error of type I is bounded above as 
\begin{align}
\bP{K_x K_y \bF Z}{\cA^c}
&\le \ttlvrn{\bPP{K_x K_y \bF Z}}{ \mathrm{P}_{\mathtt{unif}}^\two \times \bPP{\bF Z}} + 
 \mathrm{P}_{\mathtt{unif}}^\two \times \bPP{\bF Z}(\cA^c) \nonumber \\
 &\le \ep +  \mathrm{P}_{\mathtt{unif}}^\two \times \bPP{\bF Z}(\cA^c), 
  \label{eq:type-one-bound-1}
\end{align}
where the first inequality follows from the definition of variational distance, and the second is a 
consequence of the security criterion \eqref{e:security}.
The second term above can be expressed as follows:
\begin{align}
 \mathrm{P}_{\mathtt{unif}}^\two \times \bPP{\bF Z}(\cA^c)
 &= \sum_{f,z} \bP{\bF Z}{f,z} \frac{1}{|\cK|} \sum_k \indicator\left( (k,k,f,z) \in \cA^c \right) \nonumber \\
 &= \sum_{f,z} \bP{\bF Z}{f,z} \frac{1}{|\cK|} \sum_k \indicator\left( \bQ{K_x K_y|\bF Z}{k,k,|f,z} |\cK|^2 \eta^2 > 1 \right). 
  \label{eq:type-one-bound-2}
\end{align}
The inner sum can be further upper bounded as 
\begin{align}
\sum_k \indicator\left( \bQ{K_x K_y|\bF Z}{k,k,|f,z} |\cK|^2 \eta^2 > 1 \right)
&\le \sum_k \left( \bQ{K_x K_y|\bF Z}{k,k,|f,z} |\cK|^2 \eta^2 \right)^{\frac{1}{2}} \nonumber \\
&= |\cK| \eta \sum_k \bQ{K_x K_y|\bF Z}{k,k,|f,z}^{\frac{1}{2}} \nonumber \\
&= |\cK| \eta \sum_k \bQ{K_x| \bF Z}{k|f,z}^{\frac{1}{2}} \bQ{K_y|\bF Z}{k|f,z}^{\frac{1}{2}}, 
 \label{eq:type-one-bound-3}
\end{align}
where the previous equality uses Lemma \ref{lemma:property-interactive-communication} and the fact that 
given $\bF$, $K_x$ and $K_y$ are functions of $(X,U_x)$ and $(Y,U_y)$, respectively. 
Next, an application of the Cauchy-Schwartz inequality to the sum on the right-side of \eqref{eq:type-one-bound-3} yields
\begin{align}
\sum_k \bQ{K_x| \bF Z}{k|f,z}^{\frac{1}{2}} \bQ{K_y|\bF Z}{k|f,z}^{\frac{1}{2}}
&\le \left( \sum_{k_x}\bQ{K_x| \bF Z}{k_x|f,z}  \right)^{\frac{1}{2}} \left( \sum_{k_y} \bQ{K_y|\bF Z}{k_y|f,z} \right)^{\frac{1}{2}} \nonumber \\
&= 1. \label{eq:type-one-bound-4}
\end{align}
Upon combining \eqref{eq:type-one-bound-2}-\eqref{eq:type-one-bound-4}, we obtain
\begin{align*}
\mathrm{P}_{\mathtt{unif}}^\two \times \bPP{\bF Z}(\cA^c) \le \eta,
\end{align*}
which along with \eqref{eq:type-one-bound-1} gives 
\begin{align}
\bP{K_x K_y \bF Z}{\cA^c} \le \ep +  \eta. 
 \label{eq:type-one-bound-final}
\end{align}
It follows from \eqref{eq:type-one-bound-final} and \eqref{eq:type-two-bound} that 
\begin{align*}
\beta_{\ep+\eta}(\bPP{K_xK_y \bF Z}, \bQQ{K_x K_y \bF Z}) \le \frac{1}{|\cK|\eta^2},
\end{align*}
which completes the proof. \qed

Finally, we derive the upper bound for $\sleng$ using the
data processing property of $\beta_\ep$: let $W$ be a 
stochastic mapping from $\cV$ to $\cV^\prime$, i.e., for each $v \in \cV$,
$W(\cdot | v)$ is a distribution on $\cV^\prime$. Then, since the map $W$ followed by a test on $\cV^\prime$
can be regarded as a stochastic test on $\cV$,
\begin{align} \label{eq:data-processing-inequality}
\beta_\ep(\dP,\dQ) \le \beta_\ep(\dP \circ W, \dQ \circ W),
\end{align}
where $(\dP \circ W)(v^\prime) = \sum_v \bP{}{v} W(v^\prime|v)$.

\paragraph*{Proof of Theorem \ref{theorem:one-shot-converse-source-model}}
Using the data processing inequality \eqref{eq:data-processing-inequality} with
$\dP = \bPP{X Y Z}$, $\dQ = \bQQ{XYZ}$, and $W = W_{K_x K_y\bF | XYZ}$, 
Lemma \ref{lemma:relation-SK-HT} implies that any $(K_x,K_y)$ 
satisfying the secrecy criterion \eqref{e:security} must satisfy
\begin{align}
\log |\cK | \le - \log \beta_{\ep+\eta}(\bPP{X Y Z}, \bQQ{X Y  Z}) + 2 \log(1/\eta). 
\label{eq:bound-after-dpi}
\end{align}
Furthermore, from Lemma \ref{l:secret key_combined}, $(\ep,\delta)$-SK implies existence of local estimates
$K_x$ and $K_y$ satisfying \eqref{e:security} with $(\ep+\delta)$ in place of $\ep$. 
Thus, an $(\ep,\delta)$-SK with range $\cK$ 
must satisfy \eqref{eq:bound-after-dpi} with $\ep$ replaced by $(\ep+\delta)$, which 
completes the proof. \qed

%%%%%%%%%%%%%%%%%%%%%%%%%%%%%%%%%%%%%%%%%%%%%%%%
\subsection{Proof of Lemma \ref{l:leftover_hash}}
\label{appendix:l:leftover_hash}

Let $K_s = f_s(X)$ be the key for a fixed seed.
By using the Cauchy-Schwarz inequality,
\begin{align*}
\ttlvrn{\bPP{K_s VZ}}{\bPP{\mathtt{unif}}\bPP{VZ}}
&= \frac{1}{2} \sum_{k,v,z} \left| \bP{K_sVZ}{k,v,z} - \frac{1}{|\cK|} \bP{VZ}{v,z} \right| \\
&= \frac{1}{2} \sum_{k,v,z} \sqrt{\bQ{Z}{z}} \left| \frac{\bP{K_sVZ}{k,v,z} - \frac{1}{|\cK|} \bP{VZ}{v,z}}{\sqrt{\bQ{Z}{z}}} \right| \\
&\le \frac{1}{2} \sqrt{ |\cK| |\cV| \sum_{k,v,z} \frac{\left( \bP{K_sVZ}{k,v,z} - \frac{1}{|\cK|} \bP{VZ}{v,z} \right)^2 }{\bQ{Z}{z}}}.
\end{align*}
Thus, by the concavity of $\sqrt{\cdot}$, 
\begin{align*}
\ttlvrn{\bPP{K VZS}}{\bPP{\mathtt{unif}}\bPP{VZ}\bPP{S}} 
\le \frac{1}{2} \sqrt{ |\cK| |\cV| \sum_{k,v,z,s} \bP{S}{s} \frac{\left( \bP{K_sVZ}{k,v,z} - \frac{1}{|\cK|} \bP{VZ}{v,z} \right)^2 }{\bQ{Z}{z}}}.
\end{align*}
The numerator of the sum can be rewritten as 
\begin{align*}
& \sum_{k,s} \bP{S}{s} \left( \bP{K_s VZ}{k,v,z} - \frac{1}{|\cK|} \bP{VZ}{v,z} \right)^2 \\
&= \sum_s \bP{S}{s} \sum_k \bigg[ \bP{K_s VZ}{k,v,z}^2 - 2 \bP{K_s VZ}{k,v,z} \frac{1}{|\cK|} \bP{VZ}{v,z} 
  + \frac{1}{|\cK|^2}  \bP{VZ}{v,z}^2 \bigg] \\
&= \sum_s \bP{S}{s} \bigg[ \sum_k \bP{K_s VZ}{k,v,z}^2 - \frac{1}{|\cK|} \bP{VZ}{v,z}^2 \bigg] \\
&= \sum_s \bP{S}{s} \bigg[ \sum_{x,x^\prime} \bP{XVZ}{x,v,z}\bP{XVZ}{x^\prime,v,z} 
     \left\{ \indicator\left( f_s(x) = f_s(x^\prime) \right) - \frac{1}{|\cK|} \right\} \bigg] \\
&= \sum_x \bP{XVZ}{x,v,z}^2 \sum_s \bP{S}{s} \left\{ 1 - \frac{1}{|\cK|} \right\} \\
&~~~  + \sum_{x \neq x^\prime} \bP{XVZ}{x,z,v}\bP{XVZ}{x^\prime,v,z} \sum_s \bP{S}{s}
   \left\{ \indicator\left( f_s(x) = f_s(x^\prime) \right) - \frac{1}{|\cK|} \right\} \\
&\le \sum_x \bP{XVZ}{x,v,z}^2,
\end{align*}
where we used the property of two-universality \eqref{eq:property-two-universal} in the last inequality.
Thus, we have
\begin{align*}
\ttlvrn{\bPP{K VZS}}{\bPP{\mathtt{unif}}\bPP{VZ}\bPP{S}} 
&\le \frac{1}{2} \sqrt{ |\cK||\cV| \sum_{x,v,z} \frac{\bP{XVZ}{x,v,z}^2}{\bQ{Z}{z}}} \\
&\le \frac{1}{2} \sqrt{ |\cK||\cV| \sum_{x,v,z} \frac{\bP{XVZ}{x,v,z}\bP{XZ}{x,z}}{\bQ{Z}{z}} } \\
&= \frac{1}{2} \sqrt{ |\cK| |\cV| \sum_{x,z} \frac{\bP{XZ}{x,z}^2}{\bQ{Z}{z}}} \\
&\le \frac{1}{2} \sqrt{ |\cK| |\cV| 2^{- H_{\min}(\bPP{XZ}| \bQQ{Z})}}.
\end{align*}
\qed

Note that the last step in the proof above shows that it is, in fact, the conditional
R\'enyi entropy of order $2$ that determines the leakage (see \cite{BenBraCreMau95} for a similar
observation). However, the weaker bound proved above suffices for our case, as it does for
many other cases (\cf \cite{Ren05}).
%%%%%%%%%%%%%%%%%%%%%%%%%%%%%%%%%%%%%%%%%%%%%%%%
\subsection{A secret key agreement protocol requiring $1$-bit feedback}
\label{appendix:high_secrecy_protocol}

In this section, we present a secret key agreement protocol which requires only
$1$-bit of feedback for generating an $(\ep, \delta)$-SK, in the special 
case when $Z$ is a constant. The main component is a high secrecy protocol
 which achieves arbitrarily high secrecy and required reliability.
In contrast to Protocol~\ref{prot:high_reliability}, which relied on
slicing the spectrum of $\bPP{X|Y}$, the high secrecy protocol
is based on slicing the spectrum of $\bPP{X}$. Since the party
observing $X$ can determine the corresponding slice index, feedback
is not needed and one-way communication suffices. We then convert
this high secrecy protocol into a high reliability protocol, using
a $1$-bit feedback. The required protocol for generating an $(\ep, \delta)$-SK
is obtained by randomizing between the high secrecy and the high reliability protocols.
This protocol appeared in a conference version containing some of the results of this paper \cite{HayTyaWat14}, but was discovered independently by \cite{FullerSmithReyzin14} in a slightly different setting. 
Note that this is a different approach from the one used in Section~\ref{s:protocols}
where a high reliability protocol was constructed and Proposition~\ref{p:conversion}
was invoked to obtain a high secrecy protocol.

{\bf Description of the high secrecy protocol.} We now describe our protocol 
formally. The information reconciliation step of our protocol relies on
a single-shot version of the classical Slepian-Wolf theorem \cite{SleWol73}
in distributed source coding for two sources 
\cite{MiyKan95}, \cite[Lemma 7.2.1]{Han03} (see, also, \cite{Kuz12}).
We need a slight modification of the standard version -- the encoder
is still a random binning but for decoding, instead of using a ``typical-set''
decoder for the underlying distribution, we use a mismatched typical-set
decoder. We provide a proof for completeness.
%%%
\begin{lemma}[{{\bf Slepian-Wolf Coding}}]\label{l:slepian_wolf}
Given two distributions $\bPP{XY}$ and $\mathrm{Q}_{XY}$ on $\cX\times \cY$,
for every $\gamma > 0$
there exists a code $(e,d)$ of size $M$ with encoder 
$e: \cX \rightarrow \{1 ,..., M\},$
and a decoder
$d:\{1, ...,M\} \times \cY \rightarrow \cX,$
such that
\begin{align*}
\bP{XY}{\{(x,y): x \neq d(e(x), y)\}} \leq \bP{XY}{\{(x,y) : -\log \mathrm{Q}_{X\mid Y}(x\mid y) \geq \log M - \gamma\}} + 2^{-\gamma}.
\end{align*}
\end{lemma}
{\it Proof.} We use a random encoder given by {random binning}, \ie, for each $x \in \cX$, we 
independently randomly assign $i=1,\ldots,M$. 
For the decoder, we use a typicality-like argument, but instead of using the standard typical set
defined via $\bPP{X|Y}$, we use the mismatched typical-set
\begin{align*}
\cT_{\bQQ{X|Y}} := \left\{ (x,y) : - \log \bQ{X|Y}{x|y} < \log M - \gamma \right\}.
\end{align*}
Then, upon receiving $i \in \{1,\ldots,M\}$, the decoder outputs $\hat{x}$ if there exists a unique 
$\hat{x}$ satisfying $e(\hat{x}) = i$ and $(\hat{x},y) \in \cT_{\bQQ{X|Y}}$. An error occur if $(X,Y) \notin  \cT_{\bQQ{X|Y}}$ or there exists $\tilde{x} \neq X$ such that 
$(\tilde{x},Y) \in \cT_{\bQQ{X|Y}}$ and $e(\tilde{x}) = e(X)$. The former error event 
occurs with probability
\begin{align*}
\bP{XY}{\{(x,y) : -\log \mathrm{Q}_{X\mid Y}(x\mid y) \geq \log M - \gamma\}}.
\end{align*}
The probability of the second error event averaged over the random binning is bounded as
\begin{align*}
& \mathbb{E}\left[ \dP\left( \exists \tilde{x} \neq X \mbox{ s.t. } e(\tilde{x}) = e(X),~
 (\tilde{x},Y) \in \cT_{\bQQ{X|Y}} \right) \right] \\
&\le \sum_{x,y} \bP{XY}{x,y} \mathbb{E}\left[  \sum_{\tilde{x} \neq x}  \indicator\left( e(\tilde{x}) = e(x) \right) \cdot 
  \indicator\left( (\tilde{x}, y) \in \cT_{\bQQ{X|Y}} \right) \right] \\
&= \sum_{x,y} \bP{XY}{x,y}  \sum_{\tilde{x} \neq x} \frac{1}{M} \indicator\left( (\tilde{x}, y) \in \cT_{\bQQ{X|Y}} \right) \\
&\le \sum_{x,y} \bP{XY}{x,y} 2^{-\gamma} \\
&= 2^{-\gamma},
\end{align*}
where the expectation is over the random
encoder $e$. The first inequality above is by the union bound, 
the first equality is a property of random binning,
and the second inequality follows from 
\begin{align*}
|\{ x: (x,y) \in \cT_{\bQQ{X|Y}} \}| \le M 2^{-\gamma}~~~\forall y \in \cY.
\end{align*}
Thus, there exists a code $(e,d)$ satisfying the desired bound. \qed

We are now in a position to describe our protocol, which is based on
slicing the spectrum of $\bPP X$. We first {slice the spectrum} of $\bPP{X}$ into $L+1$ parts.
Specifically, for $1\leq i \leq L$, let $\lambda_i = \lambda_{\min} + (i-1)\Delta$
and define
\begin{align} \label{eq:calXi}
\cX_i := \{x : \lambda_i \leq - \log \bP{X}{x} < \lambda_{i} + \Delta\}.
\end{align}
We also define
\begin{align} \label{eq:calX0}
 \cX_0 := \left\{ (x,y) : -\log\bP{X}{x} \geq \lambda_{\max} \text{ or } -\log\bP{X}{x} < \lambda_{\min}\right\}. 
\end{align}
Denote by $J$ the RV such that the event $\{J=j\}$ corresponds to $\cX_j$, $0\leq j \le L$. 
We divide the indices $0\leq j \leq L$ into ``good" indices $\cI_g$
and the ``bad" indices $\cI_b = \cI_g^c$, where
\begin{align*}
\cI_g &= \left\{j : j>0 \text{ and }\bP{J}{j} \geq \frac{1}{L^2}\right\}.
\end{align*}
Denote by $\bPP{j}$ the conditional distribution 
of $X, Y$ given $J =j$, \ie,
\begin{align*}
\bP{j}{x, y} = \frac{\bP{XY}{x, y}}{\bP{X}{\cX_j}}\mathbf{1}(x\in \cX_j), 
\quad   x \in \cX,~y \in \cY,~ 0\leq j \leq L.
\end{align*}
Note that $J$ is a function of $X$ and can be computed by the first party, \ie,  the party 
observing $X$. In our protocol below, the first party computes $J$ and sends it to the second
party as public communication. If $J \in {\cal I}_b$, the protocol declares a reconciliation error and 
aborts. Otherwise, the protocol generates a secret key conditioned on the event $\cX_J$.

For $1\le j \le L$, let $(e_j,d_j)$ be the Slepian-Wolf code of Lemma \ref{l:slepian_wolf} 
for $\bPP{XY} = \bPP{j}$ and $\mathrm{Q}_{XY} = \bPP{XY}$.  Further, 
let $\cF$ be a {$2$-universal family} of mappings $f: \cX\rightarrow \cK$, and
let $S$ be random seed such that $f_S$ denotes a randomly chosen member of $\cF$.

Our secret key agreement protocol is given in Protocol~\ref{p:high_secrecy_Z_constant}.

\begin{protocol}[H]
\caption{High secrecy protocol}
\label{p:high_secrecy_Z_constant}
\KwIn{Observations $X$ and $Y$} 
\KwOut{Secret key estimates $K_x$ and $K_y$}
\underline{Information reconciliation}
\\
First party (observing $X$) finds the index 
$J \in \{0, 1, ..., L\}$ such that $X \in \cX_J$
\\
\eIf{$J \in \cI_b$}
{
The protocol declares an error and aborts
}
{
First party sends $(J, e_J(X))$ to the second party
\\
Second party computes $\hat X = d_J(Y, e_J(X))$
\\
\underline{Privacy amplification}
\\
First party generates the random seed $S$ and
sends it to the second party using public communication
\\
First party generates the secret key $K_x = K= f_S(X)$ 
%\begin{align*}
%K =\begin{cases}
%f_S(X),\quad &I \in \cI_g,\\
%\mathtt{unif}(\cK),\quad &\text{otherwise}.
%\end{cases}
%\end{align*}
\\
The second party generates the estimate $K_y$ of $K$
as $K_y = f_S(\hat X)$
}
\end{protocol}

{\bf Performance bounds for Protocol~\ref{p:high_secrecy_Z_constant}.} The next result shows that
Protocol~\ref{p:high_secrecy_Z_constant} attains arbitrary high secrecy and required reliability.
\begin{theorem}\label{t:protocolA1}
For every $\gamma> 0$ and $0 \leq \lambda \leq \lambda_{\min}$, Protocol~\ref{p:high_secrecy_Z_constant}
yields an $(\ep, \delta)$-SK $K$ taking values in $\cK$ with
\begin{align*}
\ep &\leq \bPr{ \idXY  \leq \lambda + \gamma +\Delta} +\bP{XY}{\cX_0} +2^{-\gamma} + \frac{1}{L},
\\
\delta &\leq \frac{1}{2}\sqrt{|\cK|2^{-(\lambda - 2\log L)}},
\end{align*}
where, with $\la_{\max} = \la_{\min}+L\Delta$, $\cX_0$ is given by \eqref{eq:calX0}.
\end{theorem}
{\it Proof.} 
To bound the error in information reconciliation, note that for all $j \in \cI_g$ 
by Lemma \ref{l:slepian_wolf} with $\bPP{XY} = \bPP{j}$ and $\mathrm{Q}_{XY} = \bPP{XY}$
\begin{align*}
&\bP{j}{\{(x,y):  x \neq d_j(e_j(x), y)\}} - 2^{-\gamma}
\\
&\leq \bP{j}{\{(x,y): -\log \mathrm{P}_{X\mid Y}(x\mid y) \geq \log M_j - \gamma\}} 
\\
&= \mathrm{P}_j
\left(\bigg\{(x,y) :\right. 
\left. - \log\bP{X}{x} - \idxy \geq \log M_j - \gamma\bigg\}\right)
\\
&\leq \mathrm{P}_j
\left(\bigg\{(x,y) :\right. 
 \left. \lambda_j + \Delta - \idxy \geq \log M_j - \gamma\bigg\}\right),
\end{align*}
where the previous inequality uses the definition of $\bPP{j}$ and \eqref{eq:calXi}. 
On choosing 
\begin{align*}
\log M_j = \lambda_j - \lambda,
\end{align*}
we get
\begin{align*}
\bP{j}{\{(x,y) : x \neq d_j(e_j(x), y)\}} 
\leq \bP{j}{\left\{(x,y) : \idxy \leq \lambda + \gamma + \Delta\right\}} + 2^{-\gamma}.
\end{align*}
An error in information reconciliation occurs if
either $J \notin \cI_g$ or if $j\in \cI_g$ and $X \neq d_j(e_j(X),Y)$.
From the bound above
\begin{align*}
\ep &\leq \bP{XY}{J\notin \cI_g}  + 2^{-\gamma}
+ 
\sum_{j \in \cI_g}\bP{J}{j}
\bP{j}{\left\{(x,y): \idxy \leq \lambda + \gamma + \Delta\right\}} 
\\
&\leq \bP{XY}{J\notin \cI_g}  + 2^{-\gamma} 
 +\bP{XY}{\left\{(x,y) : \idxy \leq \lambda + \gamma + \Delta\right\}},
\end{align*}
which using
\begin{align*}
\bP{J}{\cI_b} = \sum_{j\in \cI_b}\bP{J}{j} \leq \bP{J}{0} + \frac{1}{L}
\end{align*}
gives
\begin{align*}
\ep \leq \bP{X}{\cX_0}  + 2^{-\gamma} + \frac{1}{L}
 +\bP{XY}{\left\{(x,y): \idxy \leq \lambda + \gamma + \Delta\right\}},
\end{align*}
proving the reliability bound of the theorem. 

We proceed to secrecy analysis. Note that the protocol only defines 
the secret key for the case $J\in \cI_g$. For concreteness, let
\begin{align*}
K =\begin{cases}
f_S(X),\quad &J \in \cI_g,\\
\mathtt{unif}(\cK),\quad &\text{otherwise},
\end{cases}
\end{align*}
$K$ is perfectly secure when $J \in \cI_b$. Denoting the communication $(J, e_J(X))$ by $\bF$, we get
\begin{align}
&\ttlvrn{\bPP{K\bF S}}{\bPP{\mathtt{unif}}\bPP{\bF S}}
\nonumber
\\
&=\sum_{j\notin \cI_g}\bP{J}{j} \cdot 0  + \sum_{j\in \cI_g}\bP{J}{j} \cdot \ttlvrn{\bPP{Ke_J(X)S|J=j}}
{ \bPP{\mathtt{unif}}\bPP{e_J(X)S|J=j}}.
\label{e:secrecy_bound1}
\end{align}
To bound $\ttlvrn{\bPP{Ke_j(X_j)S|J=j}}{ \bPP{\mathtt{unif}}\bPP{e_J(X)S|J=j}}$, denote by $\bPP{j,X}$ the marginal on $\cX$ induced by $\bPP{j}$.
Note that for each $j \in \cI_g$
\begin{align*}
 - \log\bP{j,X}{x} &= -\log \frac{\bP{X}{x}}{\bP{X}{\cX_j}} 
 \\
 &\geq \lambda_j - 2\log L,
\end{align*}
where the last inequality uses the definition of $\cX_j$
and $\cI_g$. It follows that 
\begin{align*}
H_{\min}\left(\bPP{j,X}\right) \geq \lambda_j - 2\log L.
\end{align*}
Therefore, upon noting that $S$ is independent of $(X,Y,J)$ even upon conditioning 
on $J=j$, for each $j\in \cI_g$ an application of Lemma \ref{l:leftover_hash} 
implies that
\begin{align*}
\ttlvrn{\bPP{Ke_j(X)S|J=j}}{ \bPP{\mathtt{unif}}\bPP{e_J(X)S|J=j}}
&\leq \frac{1}{2} \sqrt{|\cK|M_j2^{-H_{\min}\left(\bPP{j,X}\right)}}
\nonumber
\\
&\le \frac{1}{2} \sqrt{|\cK|2^{-(\lambda - 2\log L)}}, \quad j\in \cI_g,
\end{align*}
which together with \eqref{e:secrecy_bound1} gives
\[
\ttlvrn{\bPP{K\bF S}}{\bPP{\mathtt{unif}}\bPP{\bF S}}\leq \frac{1}{2} \sqrt{|\cK|2^{-(\lambda - 2\log L)}},
\]
which in turn proves the secrecy bound claimed in the theorem.
\qed

{\bf From high secrecy protocol to a high reliability protocol.}
By Theorem~\ref{t:protocolA1}, the secrecy parameter $\delta$ of 
Protocol~\ref{p:high_secrecy_Z_constant} can be made 
small by choosing $\log |\cK|\approx \lambda$, but its reliability
parameter $\ep$ is limited by the tail-probability $\bPr{ \idXY  \leq \lambda}$.
Thus, in contrast to Protocol~\ref{prot:high_reliability},
Protocol~\ref{p:high_secrecy_Z_constant} constitutes a high
secrecy protocol. Note that while any high reliability protocol can be converted into a high secrecy protocol
using Proposition \ref{p:conversion}, it is unclear if a high secrecy protocol
can be converted to a high reliability protocol in general. However, high secrecy 
Protocol~\ref{p:high_secrecy_Z_constant} can be converted into
a high reliability protocol as follows: The second party upon decoding $X$ computes the
indicator
of the error event $\cE_j := \{-\log \bP{X\mid Y}{X\mid Y} \geq \log M_j - \gamma\}$ 
and sends it back to the first
party. If $\cE_j$ doesn't occur, the secret key $K$ is as in the protocol above. Otherwise, $K$ is chosen to be a constant. 
For this modified secret key, the event $\cE_j$ is accounted for in the secrecy parameter $\delta$ and not in $\ep$ 
as earlier. Thus, Theorem~\ref{t:protocolA1} holds for the modified secret key where
the leading term $\bPr{ \idXY  \leq \lambda + \gamma +\Delta}$ is moved from the upper bound
on $\ep$ to that on $\delta$, and the resulting protocol has 
high reliability. Furthermore, the high reliability protocol uses just $1$-bit of feedback
from the second party to the first. 

Finally, a protocol for generating an arbitrary $(\ep, \delta)$-SK can be obtained by
a hybrid use of the high reliability and high secrecy protocols as in 
Proposition~\ref{p:time_sharing}.

%%%%%%%%%%%%%%%%%%%%%%%%%%%%%%%%%%%%%%%%%%%%%%%%
\section*{Acknowledgments}
MH is partially supported by a MEXT Grant-in-Aid for Scientific Research (A) No. 23246071. 
MH is also partially supported by the National Institute of
Information and Communication Technology (NICT), Japan.
The Centre for Quantum Technologies is funded by the Singapore Ministry
of Education and the National Research Foundation as part of
the Research Centres of Excellence program.
%%%%%%%%%%%%%%%%%%%%%%%%%%%%%%%%%%%%%%%%%%%%%%%
\bibliography{IEEEabrv,references}
\bibliographystyle{IEEEtranS}

 %%%%%%%%%%%%%%%%%%%%%%%%%%%%%%%%%%%%%%%%%
\end{document}